\newif\ifabstract
\abstracttrue
\newif\iffull
\ifabstract \fullfalse \else \fulltrue \fi

\documentclass[11pt]{article}
\usepackage{amsfonts}
\usepackage{amssymb}
\usepackage{amstext}
\usepackage{amsmath}
\usepackage{xspace}
\usepackage{theorem}
\usepackage{graphicx}
\usepackage{graphics}
\usepackage{colordvi}
\usepackage{colordvi}

\advance \oddsidemargin by -0.8in
\newcommand{\myparskip}{3pt}
\parskip \myparskip

        {\hspace*{\fill}$\Box$\par\vspace{4mm}}

\newcommand{\qed}{\hfill\vbox{\hrule height.2pt\hbox{\vrule width.2pt height5pt \kern5pt
\vrule width.2pt} \hrule height.2pt}}

\newcommand{\connect}{\leadsto}

\newcommand{\sconnect}{\overset{\mbox{\tiny{1:1}}}{\leadsto}}
\newcommand{\paths}[4]{\ensuremath{#1: #2 \leadsto_{#4} #3}}
\newcommand{\flow}[4]{\ensuremath{#1: #2 \leadsto_{#4} #3}}
\newcommand{\tH}{\tilde{H}}

\newcommand{\optLP}{\mathsf{OPT}_{\mathsf{LP}}}

\newcommand{\alphawl}{\ensuremath{\alpha_{\mbox{\tiny{\sc WL}}}}}
\newcommand{\alphaWL}{\alphawl}
\newcommand{\algsc}{\ensuremath{{\mathcal{A}}_{\mbox{\textup{\scriptsize{ARV}}}}}\xspace}
\newcommand{\algKRV}{\ensuremath{{\mathcal{A}}_{\mbox{\textup{\scriptsize{KRV}}}}}\xspace}

\newcommand{\alphasc}{\ensuremath{\alpha_{\mbox{\tiny{\sc ARV}}}}}
\newcommand{\gkrv}{\ensuremath{\gamma_{\mbox{\tiny{\sc KRV}}}}}

\newcommand{\event}{{\cal{E}}}

\newcommand{\NP}{\mbox{\sf NP}\xspace}

\newcommand{\opt}{\mathsf{OPT}}
\newcommand{\te}{\tilde e}

\newcommand{\set}[1]{\left\{ #1 \right\}}
\newcommand{\sse}{\subseteq}

\newcommand{\tset}{{\mathcal T}}
\newcommand{\eset}{{\mathcal E}}

\newcommand{\pset}{{\mathcal{P}}}
\newcommand{\qset}{{\mathcal{Q}}}

\newcommand{\aset}{{\mathcal{A}}}
\newcommand{\cset}{{\mathcal{C}}}
\newcommand{\fset}{{\mathcal{F}}}
\newcommand{\mset}{{\mathcal M}}

\newcommand{\wset}{{\mathcal{W}}}

\newcommand{\gset}{{\mathcal{G}}}

\newcommand{\tV}{\tilde{V}}

\newcommand{\nots}{\overline S}

\newcommand{\be}{\begin{enumerate}}
\newcommand{\ee}{\end{enumerate}}
\newcommand{\bd}{\begin{description}}
\newcommand{\ed}{\end{description}}
\newcommand{\bi}{\begin{itemize}}
\newcommand{\ei}{\end{itemize}}

\newtheorem{theorem}{Theorem}
\newtheorem{observation}{Observation}
\newtheorem{corollary}{Corollary}
\newtheorem{claim}{Claim}

\newtheorem{definition}{Definition}
\newenvironment{proof}{\par \smallskip{\bf Proof:}}{\hfill\stopproof}
\def\stopproof{\square}
\def\square{\vbox{\hrule height.2pt\hbox{\vrule width.2pt height5pt \kern5pt
\vrule width.2pt} \hrule height.2pt}}

\renewcommand{\phi}{\varphi}
\newcommand{\eps}{\epsilon}

\newcommand{\half}{\ensuremath{\frac{1}{2}}}

\newcommand{\poly}{\operatorname{poly}}

\newcommand{\reals}{{\mathbb R}}

\newcommand{\prob}[2][]{\text{\bf Pr}_{#1}\left [#2\right]}

\newcommand{\EDP}{\mbox{\sf EDP}\xspace}
\newcommand{\EDPwC}{\mbox{\sf EDPwC}\xspace}

\newcommand{\dmax}{d_{\mbox{\textup{\footnotesize{max}}}}}

\newcommand{\out}{\operatorname{out}}

\begin{document}

\title{Routing in Undirected Graphs with Constant Congestion}
\author{Julia Chuzhoy\thanks{Toyota Technological Institute, Chicago, IL
60637. Email: {\tt cjulia@ttic.edu}. Supported in part by NSF CAREER grant CCF-0844872 and Sloan Research Fellowship.}}

\begin{titlepage}
\maketitle

\thispagestyle{empty}
\begin{abstract}
Given an undirected graph $G=(V,E)$, a collection $(s_1,t_1),\ldots,(s_k,t_k)$ of $k$ source-sink pairs, and an integer $c$, the goal in the Edge Disjoint Paths with Congestion problem is to connect  maximum possible number of the source-sink pairs by paths, so that the maximum load on any edge (called edge congestion) does not exceed $c$.

We show an efficient randomized algorithm to route $\Omega(\opt/\poly\log k)$ source-sink pairs with congestion at most 14, where $\opt$ is the maximum number of pairs that can be simultaneously routed on edge-disjoint paths.
The best previous algorithm that routed $\Omega(\opt/\poly\log n)$ pairs  required congestion $\poly(\log \log n)$, and for the setting where the maximum allowed congestion is bounded by a constant $c$, the best previous algorithms could only guarantee the routing of $\opt/n^{O(1/c)}$ pairs.
\end{abstract}
\end{titlepage}


\section{Introduction}
We study network routing problems in undirected graphs. In such problems, we are given an undirected $n$-vertex graph $G=(V,E)$, and a collection $\mset=\set{(s_1,t_1),(s_2,t_2),\ldots,(s_k,t_k)}$ of $k$ \emph{source-sink pairs}, that we also refer to as \emph{demand pairs}. In order to route a pair $(s_i,t_i)$, we need to select a path connecting $s_i$ to $t_i$ in graph $G$. Given a routing of any subset of the demand pairs, its \emph{congestion} is the maximum load on any edge, that is, the maximum number of paths containing the same edge. In general, we would like to route as many demand pairs as possible, while minimizing the edge congestion. These two conflicting objectives naturally give rise to a number of basic optimization problems. 

One of the central routing problems is Edge Disjoint Paths (\EDP), where the goal is to route the maximum number of demand pairs on edge-disjoint paths (that is, with congestion $1$). Robertson and Seymour~\cite{RobertsonS} have shown an efficient algorithm to solve this problem, when the number $k$ of the demand pairs is bounded by a constant. However, for general values of $k$, it is NP-hard to even decide whether all pairs can be simultaneously routed on edge-disjoint paths~\cite{Karp}. The best currently known approximation algorithm for the problem, due to Chekuri, Khanna and Shepherd~\cite{EDP-alg}, achieves an $O(\sqrt{n})$-approximation factor, while the best current hardness of approximation is $\Omega(\log^{1/2-\epsilon}n)$ for any constant $\eps$, unless \NP has randomized quasi-polynomial time algorithms~\cite{AZ-undir-EDP,ACGKTZ}. We note that the standard multicommodity flow LP relaxation for \EDP, which is commonly used in approximation algorithms for network routing problems, has an integrality gap of $\Omega(\sqrt n)$~\cite{EDP-alg}. Interestingly, Rao and Zhou~\cite{RaoZhou} have shown a factor $\poly\log n$-approximation for \EDP on graphs where the value of the global minimum cut is $\Omega(\log^5n)$, by rounding the same LP relaxation.

On the other extreme is the Congestion Minimization problem, where we are required to route all source-sink pairs, while minimizing the edge congestion. The classical randomized rounding technique of Raghavan and Thompson~\cite{RaghavanT} gives the best currently known approximation algorithm for this problem, whose approximation factor is $O(\log n/\log\log n)$. On the negative side, Andrews and Zhang~\cite{AZ-undir-cong} show that the problem is hard to approximate to within a factor of $\Omega\left(\frac{\log\log n}{\log\log \log n}\right )$ unless \NP has randomized quasi-polynomial time algorithms.

A problem that lies between these two extremes, and is a natural framework for studying the tradeoff between the number of pairs routed and the edge congestion is the Edge Disjoint Paths with Congestion problem (\EDPwC). We say that an algorithm $\aset$ achieves a factor $\alpha$-approximation with congestion $c$ for \EDPwC, iff it routes at least $\opt/\alpha$ of the source-sink pairs, and the congestion of this routing is bounded by $c$, where $\opt$ is the maximum number of demand pairs that can be simultaneously routed on edge-disjoint paths. In particular, a very interesting question is whether, by slightly relaxing the conditions of the \EDP problem, namely allowing a small edge congestion, we can significantly increase the number of pairs routed. 

When the  congestion $c$ is allowed to be as high as $\Omega(\log n/\log\log n)$, the randomized rounding algorithm of Raghavan and Thompson~\cite{RaghavanT} gives a constant factor approximation for \EDPwC. For smaller values of $c$, until recently, only $O(n^{1/c})$-approximation algorithms have been known~\cite{AzarR, BavejaS, KolliopoulosS}. In a recent breakthrough, Andrews~\cite{Andrews} has shown a randomized algorithm to route $\Omega\left (\frac{\opt}{\log^{61}n}\right )$ pairs with congestion $O((\log\log n)^6)$. In another recent result, Kawarabayashi and Kobayashi~\cite{KawarabayashiK} have shown an algorithm that routes $\Omega\left (\frac{\opt}{n^{3/7}}\right )$ pairs with congestion $2$, thus improving the best previously known $O(\sqrt n)$-approximation for $c=2$. 

In this paper we show an efficient randomized algorithm, that routes $\Omega\left(\frac{\opt}{\log^{23.5}k\log\log k}\right )$ demand pairs with congestion at most $14$. We note that on the negative side, Andrews et al.~\cite{ACGKTZ} have shown that for any constant $\eps$, for any $1\leq c\leq O\left(\frac{\log \log n}{\log\log\log n}\right )$, there is no $O\left ((\log n)^{\frac{1-\eps}{c+1}}\right)$-approximation algorithm for \EDPwC with congestion $c$, unless \NP has randomized quasi-polynomial time algorithms. Therefore, the best approximation factor one may hope to achieve for \EDPwC in the setting where the maximum allowed congestion is bounded by a constant is polylogarithmic.

\paragraph{Other related results}
\EDP and its variants have been studied extensively, and better approximation algorithms are known for several special cases. Some examples include planar graphs~\cite{Frank-planar, KT-planar,Kleinberg-planar,CKS, CKS-planar,Kplanar}, trees~\cite{trees2,trees1}, and expander graphs~\cite{LR,BFU,BFSU,KleinbergR,Frieze}. 

We note that routing problems are somewhat better understood in directed graphs. The EDP problem has $\tilde{O}\left(\min\set{n^{2/3},\sqrt m}\right )$-approximation algorithms in directed graphs, where $m$ is the number of graph edges~\cite{dir-EDP1,dir-EDP2,dir-EDP3}, and it is hard to approximate to within a factor of $\Omega\left (m^{1/2-\epsilon}\right)$ for any constant $\eps$~\cite{dir-EDP-hardness}. The randomized rounding technique of Raghavan and Thompson~\cite{RaghavanT} gives an $O(\log n/\log\log n)$-approximation for directed Congestion Minimization, and the problem is hard to approximate to within a factor of $\Omega(\log n/\log\log n)$~\cite{AZ-dir-cong,CGKT}, unless \NP has randomized quasi-polynomial time algorithms. As for \EDPwC,  the randomized rounding technique gives an $O(cn^{1/c})$-approximation~\cite{KolliopoulosS,Srinivasan} for any congestion bound $c$, and for any $1\leq c\leq O\left (\frac{\log n}{\log\log n}\right )$, there is no $n^{\Omega(1/c)}$-approximation algorithm for the problem unless \NP has randomized quasi-polynomial time algorithms~\cite{CGKT}.

\paragraph{Our results and techniques}
Our main result is summarized in the following theorem.

\begin{theorem}\label{thm: main}
There is a randomized polynomial-time algorithm, that, given a graph $G$ and a set $\mset=\set{(s_1,t_1),\ldots,(s_k,t_k)}$ of $k$ demand pairs, w.h.p. finds a collection $\pset$ of paths, connecting $\Omega\left (\frac{\opt}{\log^{23.5}k\log\log k}\right )$ of the demand pairs with congestion at most $14$, where $\opt$ is the maximum number of the demand pairs that can be simultaneously routed on edge-disjoint paths in $G$.
\end{theorem}

Our algorithm in fact routes $\Omega\left (\frac{\optLP}{\log^{23.5}k\log\log k}\right )$ demand pairs, where $\optLP$ is the value of the optimal solution to the standard multicommodity flow linear programming relaxation for the problem. Since the integrality gap of this LP relaxation is $\Omega(\sqrt{n})$ for EDP (that is, when no congestion is allowed), our result shows that the integrality gap improves from polynomial to polylogarithmic if we allow a congestion of $14$.  

We now give an overview of our techniques and compare them to previous work. One of the basic notions used throughout the algorithm is that of well-linkedness. Informally, given a graph $G=(V,E)$ with a set $\tset\sse V$ of vertices called terminals, we say that $G=(V,E)$ is well-linked for $\tset$, iff any set $D$ of demands on the terminals (where each terminal participates in at most one demand pair) can be routed with low congestion in $G$ (we make this definition precise in the following sections). For any subset $S\sse V$ of vertices, let $\out(S)$ be the subset of edges with exactly one endpoint in $S$. Given a subset $S\sse V$ of vertices of $G$, we say that $S$ is well-linked, iff it is well-linked for the set $\out(S)$ of edges (in order to obtain the standard definition of well-linkedness, subdivide each edge $e\in \out(S)$ with a terminal $t_e$, and consider the sub-graph of the resulting graph induced by $S\cup\tset$ with the set $\tset$ of terminals, where $\tset=\set{t_e\mid e\in \out(S)}$).
The starting point of our algorithm is the same as in the work of~\cite{CKS,RaoZhou,Andrews}. We start with the standard LP-relaxation for the \EDP problem on graph $G$, and we compute a partition of $G$ into disjoint induced sub-graphs $G_1,\ldots,G_r$. For each $1\leq i\leq r$, we compute a subset $\mset_i\sse\mset$ of demand pairs that are contained in $G_i$, such that the graph $G_i$ is well-linked for the corresponding set $\tset_i$ of terminals, containing all vertices that participate in the pairs in $\mset_i$, and moreover, $\sum_{i=1}^r|\mset_i|\geq \Omega\left (\frac{|\mset|}{\log^2k}\right )$. An algorithm for efficiently computing such a decomposition was shown by Chekuri, Khanna and Shepherd~\cite{CKS}. From now on, it is enough to find a good routing in each resulting sub-instance $G_i$ separately. To simplify notation, let $G$ denote any such sub-instance $G_i$, let $\mset$ denote the set $\mset_i$ of demand pairs, and let $\tset$ denote the corresponding set $\tset_i$ of terminals. Since graph $G$ is well-linked for $\tset$, it has good expansion properties with respect to $\tset$. However, graph $G$ may be far from being an expander, since it may contain many vertices besides the terminals. Intuitively, a natural approach is to embed an expander $X$, whose vertex set is $\tset$, into the graph $G$. Each edge $e=(t_i,t_j)$ of the expander is mapped to a path $P_e$ connecting $t_i$ to $t_j$ in $G$, and the congestion of the embedding is the maximum, over all edges $e'\in E(G)$, of the number of paths in $\set{P_e\mid e\in E(X)}$, containing $e'$. If we could find a low-congestion embedding of an expander $X$ into $G$, then we could use existing algorithms for routing on expanders to find a low-congestion routing of a polylogarithmic fraction of the demand pairs in $X$, which in turn would give us a low-congestion routing of the same demand pairs in $G$. This is the approach that has been used by Rao and Zhou~\cite{RaoZhou} and by Andrews~\cite{Andrews}. A very useful tool in embedding an expander into any well-linked graph is the cut-matching game of Khandekar, Rao and Vazirani~\cite{KRV}. In this game, we have two players: a cut player and a matching player. The cut player wants to construct an expander $X$, while the matching player tries to delay its construction. We start with $X$ containing only the set $V(X)$ of $2N$ vertices and no edges. In each iteration $i$, the cut player computes a partition $(A_i,B_i)$ of $V(X)$ with $|A_i|=|B_i|=N$, and the matching player computes a matching $M_i$ between $A_i$ and $B_i$. The edges of $M_i$ are then added to $X$. Khandekar, Rao and Vazirani~\cite{KRV} have shown that no matter what the matching player does, there is a strategy for the cut player (that we denote by $\algKRV$), such that after $O(\log^2N)$ iterations, $X$ becomes an expander. A natural strategy for constructing an expander $X$ and embedding it into graph $G$ using the cut-matching game, is the following. We will use the algorithm $\algKRV$ for the cut player, while the matching player will be simulated by finding appropriate flows in graph $G$. Specifically, we let $V(X)=\tset$ be the set of vertices of $X$. If $(A_i,B_i)$ is the bi-partition of $V(X)$ computed by the cut player, then we can try to send $|A_i|=|B_i|$ flow units from the terminals of $A_i$ to the terminals of $B_i$ in graph $G$, and use the resulting flow to define the matching $M_i$. This procedure can be used to both construct the expander $X$, and  embed it into the graph $G$. In fact, Khandekar, Rao and Vazirani use precisely this procedure in their algorithm for the sparsest cut problem.

One problem with this approach is that we need to compute $\Theta(\log^2k)$ different flows in graph $G$, and together they may cause a poly-logarithmic congestion. Moreover, the partitions that the cut player computes depend on the matchings returned by the matching player in previous iterations, so we cannot attempt to route all these flows simultaneously in graph $G$ with low congestion. Rao and Zhou~\cite{RaoZhou} have proposed the following approach to overcome this difficulty. Let $\gamma=\Theta(\log^2k)$ be the number iterations in the  algorithm of~\cite{KRV}. We can build $\gamma$ graphs $G_1,\ldots,G_{\gamma}$, where for each $1\leq i\leq \gamma$, $V(G_i)=V(G)$, and the sets $E(G_1),\ldots,E(G_{\gamma})$ of edges form a partition of the edges in $E(G)$. If we can construct the family $G_1,\ldots,G_{\gamma}$ of graphs in such a way that each graph $G_i$ is still well-linked for the terminals, then we can now construct the expander $X$ and embed it into $G$ by using the cut-matching game of \cite{KRV}, where for each $1\leq i\leq \gamma$, in order to construct the matching $M_i$, we find a flow connecting vertices of $A_i$ to vertices of $B_i$ in the graph $G_i$. Since the edges in each set $M_i$ are embedded into distinct graphs $G_i$, the congestion does not accumulate, and we obtain a good embedding of $X$ into $G$. In order to construct the graphs $G_i$, Rao and Zhou use a random procedure where each edge $e\in E$ is added to one of the graphs $G_i$ uniformly at random. However, in order to ensure that each resulting graph $G_i$ is w.h.p. well-linked 
for the terminals, the initial graph $G$ must have a large global minimum cut, that is, the value of the global minimum cut in $G$ must be at least poly-logarithmic. In order to overcome this difficulty, Andrews~\cite{Andrews} uses Raecke's tree decomposition technique~\cite{Raecke}.
Roughly speaking, he decomposes the graph $G$ into a collection $\cset$ of disjoint clusters, where each cluster $C\in \cset$ is well-linked. Moreover, if $H$ is the graph obtained from $G$ by contracting each cluster $C\in \cset$ into a single vertex, then $H$ is both well-linked for the terminals, and the value of the global minimum cut in $H$ is high enough, so we can use the algorithm of Rao and Zhou to complete the routing.

Our algorithm uses a slightly different notion of embedding one graph into another. Specifically, we still construct an expander $X$ on a subset $\tset'\sse \tset$ of terminals, but we embed it into the graph $G$ differently. Each vertex $t\in V(X)$ is represented by a connected component $C_t$ in graph $G$, that contains the terminal $t$. Each edge $e=(t,t')\in E(X)$ is represented by a path $P_e$ connecting some vertex $v\in C_t$ to some vertex $v'\in C_{t'}$ in graph $G$. Moreover, we ensure that each edge $e'\in E(G)$ only participates in a constant number of the connected components $\set{C_t}_{t\in V(X)}$ and paths $\set{P_e}_{e\in E(X)}$. Once we find such an embedding, we use {\bf vertex-disjoint} routing in the expander $X$, that gives a low edge-congestion routing in the original graph $G$. Since we construct the expander $X$ using the procedure of~\cite{KRV}, the degree of every vertex in $X$ is bounded by $O(\log^2k)$, and so a good routing on vertex-disjoint paths can be found in  $X$ using standard algorithms for routing in expanders.

A major point of our departure from previous work is how the expander $X$ is constructed and embedded into $G$. A central notion in our algorithm is that of a good family of vertex sets. Let $k'=k/\poly\log k$ be some parameter, where $k=|\mset|$ is the number of the demand pairs. We say that a subset $S\sse G$ of vertices is a \emph{good subset}, iff there is a collection $\Gamma\sse \out(S)$ of $k'$ edges, such that $G[S]\cup \Gamma$ is well-linked for $\Gamma$, and moreover the edges in $\Gamma$ can send $|\Gamma|$ flow units to the terminals in $\tset$ with low edge-congestion in graph $G$. A family $\fset$ of vertex subsets is a \emph{good family} iff it contains $\gamma$ mutually disjoint good vertex subsets $S_1,\ldots,S_{\gamma}$, where $\gamma=O(\log^2k)$ is the parameter from the cut-matching game of~\cite{KRV}.

Suppose we have found a good family $\fset=\set{S_1,\ldots,S_{\gamma}}$ of vertex subsets. For each $1\leq j\leq \gamma$, let $\Gamma_j\sse \out(S_j)$ be the corresponding subset $\Gamma$ of edges. In order to construct the expander $X$, we select a subset $\tset'=\set{t_1,\ldots,t_{k'}}\sse \tset$ of $k'$ terminals, and  we let $V(X)=\tset'$. For each $1\leq i\leq k'$, we then construct a connected component $C_i$ in graph $G$, that contains, for each $1\leq j\leq \gamma$, a distinct edge $e_{i,j}\in \Gamma_j$, and also contains the terminal $t_i$. For each $1\leq j\leq \gamma$, the edges $e_{1,j},\ldots,e_{k',j}$ are all distinct, and we view the edge $e_{i,j}$ as the copy of terminal $t_i$ for the set $S_j$. We also ensure that each edge of graph $G$ only participates in a constant number of components $\set{C_i}_{i=1}^{k'}$. For each $i$, $C_i$ is viewed as representing the vertex $t_i$ of $X$ in graph $G$. In order to construct the expander $X$, we use the cut-matching game of~\cite{KRV}, where in each iteration $1\leq j\leq \gamma$, we use the sub-graph $G[S_j]$ to route some matching $M_j$ between the copies of the terminals in $A_j$ and the terminals in $B_j$ for set $S_j$. This ensures that the congestion does not accumulate across different iterations.

Finally, we show an efficient algorithm for finding a good family $\fset$ of vertex subsets. Throughout the algorithm, we maintain a \emph{contracted graph} $G'$. This graph is uniquely defined by a collection $\cset$ of disjoint vertex subsets of $V(G)$. Given the set $\cset$ of clusters, graph $G'$ is obtained from $G$ by contracting every cluster $C\in \cset$ into a vertex $v_C$. We say that $G'$ is a \emph{legal contracted graph} for $G$ iff for each $C\in \cset$, $|\out(C)|\leq k'$, cluster $C$ is well-linked for $\out(C)$, and it does not contain any terminals. We show a randomized algorithm, that, given a legal contracted graph $G'$, w.h.p. either finds a good family $\fset$ of vertex subsets in the original graph $G$, or returns a new legal contracted graph $G''$ with $|E(G'')|<|E(G')|$. Therefore, after at most $|E(G)|$ such iterations, our algorithm is guaranteed to return a good family of vertex subsets w.h.p. Each such iteration is executed as follows. Given a current contracted graph $G'$, we find a random partition $V_1,\ldots,V_{\gamma}$ of its vertices. For each $1\leq j\leq \gamma$, we then try to recover a good vertex subset $S_j\sse V(G)$ from the set $V_j\sse V(G')$ of vertices. If we succeed to do so for all $1\leq j\leq \gamma$, then we have found a good family of vertex subsets. Otherwise, if we fail to recover a good vertex subset for some $1\leq j\leq \gamma$, then we find a new legal contracted graph $G''$ that contains strictly fewer edges than $G'$. The heart of this algorithm is a (somewhat non-standard) well-linked decomposition procedure, that is applied to each set $V_j$ in turn, and whose result is either a good subset $S_j$ of vertices, or a new contracted graph $G''$.

\section{Preliminaries and Notation}
\label{--------------------------------sec: prelims--------------------------------------}
\label{sec: Prelims}
\paragraph{Problem Definition}
We are given an undirected graph $G=(V,E)$, and a set $\mset=\set{(s_1,t_1),\ldots,(s_k,t_k)}$ of $k$ source-sink pairs, that we also refer to as demand pairs. We denote by $\tset$ the set of vertices that participate in pairs in $\mset$, and we call them \emph{terminals}. Let $\opt$ denote the maximum number of demand pairs that can be simultaneously routed via edge-disjoint paths. Our goal is to find a collection of $\Omega(\opt/\poly\log k)$ paths connecting distinct source-sink pairs, with congestion at most $14$.

We assume w.l.o.g. that each terminal in $\tset$ participates in exactly one source-sink pair. Otherwise, if a terminal $v\in \tset$ participates in $r>1$ source-sink pairs, we can add $r$ new terminals $t_1(v),\ldots,t_r(v)$, connect each of them to $v$ with an edge, and use a distinct terminal in $\set{t_1(v),\ldots,t_r(v)}$ for each source-sink pair in which $v$ participates.

We also assume w.l.o.g. that the maximum vertex degree in $G$ is $4$, and that the degree of every terminal is $1$. In order to achieve this, we perform the following simple transformation to graph $G$. First, if $v$ is a terminal, whose degree is greater than $1$, then we add a new vertex $u$ to graph $G$ that connects to $v$ with an edge, and becomes a terminal instead of $v$. Next, we process the non-terminal vertices one-by-one. Let $v$ be any such vertex, and assume that the degree of $v$ is $d>4$. Let $u_1,\ldots,u_d$ be the vertices that are neighbors of $v$. We replace $v$ with a $d\times d$ grid $Z_v$, and denote by $u_1',\ldots,u_d'$ the vertices in the first row of $Z_v$. For each $1\leq i\leq d$, we add an edge $(u_i,u'_i)$. It is easy to verify that any solution to the EDP problem in the original graph can be transformed into a feasible routing of the same value and no congestion in the new graph, and any routing in the new graph with congestion $\eta$ can be transformed into a routing in the original graph with the same congestion. Therefore, we assume from now on that the maximum vertex degree in $G$ is $4$, the degree of every terminal is $1$, and every terminal participates in one source-sink pair. 

\paragraph{General Notation}
For a graph $G=(V,E)$, and subsets $V'\sse V$, $E'\sse E$ of its vertices and edges respectively, we denote by $G[V']$, $G\setminus V'$, and $G\setminus E'$ the sub-graphs of $G$ induced by $V'$, $V\setminus V'$, and $E\setminus E'$, respectively.
For any subset $S\sse V$ of vertices, we denote by $\out_G(S)=E_G(S,V\setminus S)$ the subset of edges with  one endpoint in $S$ and the other endpoint in $V\setminus S$, and by $E_G(S)$ the subset of edges with both endpoints in $S$. When clear from context, we omit the subscript $G$. Throughout the paper, we say that a random event succeeds w.h.p., if the probability of success is $(1-1/\poly(n))$, where $n$ is the number of vertices in the input graph.

Let $\pset$ be any collection of paths in graph $G$. We say that paths in $\pset$ cause congestion $\eta$ in $G$, iff the maximum number of paths in $\pset$ containing any edge is $\eta$. In other words, if $\pset_e\sse \pset$ is the subset of paths that contain the edge $e\in E(G)$, then $\max_{e\in E(G)}\set{|\pset_e|}=\eta$.

\begin{definition}
Assume that we are given a subset $S\sse V$ of vertices and a subset $E'\sse E$ of edges of $G$. We say that a collection $\pset$ of paths connects the vertices of $S$ to the edges of $E'$ with congestion $\eta$, and denote $\paths{\pset}{S}{E'}{\eta}$, iff $\pset=\set{P_v\mid v\in S}$, where path $P_v$ has $v$ as its first vertex and some edge of $E'$ as its last edge, and $\pset$ causes congestion at most $\eta$ in $G$. In particular, each edge in $E'$ serves as the last edge on at most $\eta$ paths in $\pset$. Similarly, given two subsets $S,S'$ of vertices, if $\pset$ is a collection of paths, connecting every vertex in $S$ to some vertex in $S'$ with overall congestion at most $\eta$, then we denote this by $\paths{\pset}{S}{S'}{\eta}$. Finally, if $|S|=|S'|=|\pset|$, and each path in $\pset$ connects a distinct vertex of $S$ to a distinct vertex of $S'$, then we denote this by $\pset:S\sconnect_{\eta} S'$
\end{definition}

Similarly to the above definition, we say that a flow $F$ connects the vertices of $S$ to the edges of $E'$ with congestion $\eta$, and denote  $\flow{F}{S}{E'}{\eta}$, iff each vertex $v\in S$ sends one flow unit to the edges in $E'$, and the flow $F$ causes congestion at most $\eta$ in $G$. Notice that each flow-path in $F$ starts at a vertex of $S$ and terminates at some edge $e\in E'$. We view edge $e$ as part of the flow-path, so in particular each edge in $E'$ receives at most $\eta$ flow units. Notice that from the integrality of flow, there is a flow $\flow{F}{S}{E'}{\eta}$ iff there is a collection  $\paths{\pset}{S}{E'}{\eta}$ of paths.

Given a graph $G=(V,E)$, and a subset $\tset\sse V$ of terminals, a set $D$ of demands is a function $D:\tset\times \tset\rightarrow \reals^+$, that specifies, for each pair $t,t'\in \tset$ of terminals, a demand $D_{t,t'}$. For simplicity, we assume that the pairs $t,t'$ of terminals are unordered, that is $D_{t,t'}=D_{t',t}$ for all $t,t'\in\tset$.
We say that the set $D$ of demands is $\gamma$-restricted, iff for each $t\in \tset$, the total demand $\sum_{t'\in \tset}D_{t,t'}\leq \gamma$.
We say that the set $D$ of demands is \emph{integral} iff $D_{t,t'}$ is an integer for each $t,t'\in \tset$.

Given any set $D$ of demands, a \emph{fractional routing} of $D$ is a flow $F$, where for each unordered pair $t,t'\in \tset$, the amount of flow sent from $t$ to $t'$ (or from $t'$ to $t$) is $D_{t,t'}$. Given an integral set $D$ of demands, an \emph{integral} routing of $D$ is a collection $\pset$ of paths, where for each unordered pair $(t,t')\in \tset$, there are $D_{t,t'}$ paths connecting $t$ to $t'$ in $\pset$. The congestion of this integral routing is the congestion caused by the set $\pset$ of paths in $G$.

Given a matching $\mset$ on a set $\tset$ of vertices, we say that $\mset$ can be routed in $G$ with congestion $\eta$, iff the set $D$ of demands, where $D_{t,t'}=1$ iff $(t,t')\in \mset$, and $D_{t,t'}=0$ otherwise, can be routed in $G$ with congestion at most $\eta$.

\paragraph{Sparsest Cut and the Flow-Cut Gap} 
Suppose we are given a graph $G=(V,E)$, with non-negative weights $w_v$ on vertices $v\in V$, and a subset $\tset\sse V$ of $k$ terminals, such that for all $v\not\in \tset$, $w_v=0$. 
For any subset $S\sse V$ of vertices, let $w(S)=\sum_{v\in S}w(v)$. The sparsity of a cut $(S,\nots)$ in $G$ is $\Phi(S)=\frac{|E(S,\nots)|}{w(S)\cdot w(\nots)}$, and the value of the sparsest cut in graph $G$ is defined to be:
$\Phi(G)=\min_{S\subset V}\set{\Phi(S)}$.
In the sparsest cut problem, the input is a graph $G$ with non-negative weights on vertices, and the goal is to find a cut of minimum sparsity. Arora, Rao and Vazirani~\cite{ARV} have shown an $O(\sqrt {\log k})$-approximation algorithm for the sparsest cut problem. We will often work with a special case of the sparsest cut problem, where for each $t\in \tset$, $w_t=1$.

A problem dual to the sparsest cut problem is the maximum concurrent flow problem. For the case where the weights of all terminals are unit, the goal in the maximum concurrent flow problem is to find the maximum value $\lambda$, such each pair of terminals can send $\lambda$ flow units to each other with no congestion. The flow-cut gap is the maximum possible ratio, in any graph, between the value of the minimum sparsest cut and the maximum concurrent flow. The value of the flow-cut gap in undirected graphs, that we denote by $\beta(k)$ throughout the paper, is $\Theta(\log k)$~\cite{LR, GVY,LLR,Aumann-Rabani}. Therefore, if $\Phi(G)=\alpha$, then every pair of terminals can send $\alpha/\beta(k)$ flow units to each other with no congestion.

We will use a sightly different, but also standard, and roughly equivalent, definition of sparsity.
Given any partition $(S,\nots)$ of $V$, the \emph{sparsity} of the cut $(S,\nots)$ is $\Psi(S,\nots)=\frac{|E(S,\nots)|}{\min\set{w(S),w(\nots)}}$. We then denote:
$\Psi(G)=\min_{S\subset V}\set{\Psi(S,\nots)}$.

It is easy to see that $2\Psi(G)/k\geq \Phi(G)\geq \Psi(G)/k$. Therefore, if $\Psi(G)=\alpha$, then $\Phi(G)\geq \alpha/k$, and every pair of terminals can send $\frac{\alpha}{k\beta(k)}$ flow units to each other with no congestion. Equivalently, every pair of terminals can send $1/k$ flow units to each other with congestion at most $\beta(k)/\alpha$. Moreover, any matching on the set $\tset$ of terminals can be fractionally routed with congestion at most $2\beta(k)/\alpha$.
In the rest of the paper, we will use the latter definition of sparsity, and we will use the term cut sparsity and the value of sparsest cut to denote $\Psi(S,\nots)$ and $\Psi(G)$ respectively. The algorithm of Arora, Rao and Vazirani~\cite{ARV} can still be used to obtain a cut in graph $G$ whose sparsity is at most $O(\sqrt{\log k})\cdot \Psi(G)$.
We denote by $\algsc$ this algorithm and by $\alphasc(k)=O(\sqrt{\log k})$ its approximation factor.

\paragraph{Routing on Expanders}
\begin{definition} We say that a graph $G=(V,E)$ is an $\alpha$-expander, iff
$\min_{\stackrel{S\sse V:}{|S|\leq |V|/2}}\set{\frac{|E(S,\nots)|}{|S|}}\geq \alpha$
\end{definition}


There are many algorithms for routing on expanders, e.g. \cite{LR,BFU,BFSU,KleinbergR,Frieze}, which give different types of guarantees. For example, Frieze~\cite{Frieze} has shown that if $G$ is an $r$-regular graph (where $r$ is a constant) with strong enough expansion properties, then there is an efficient randomized algorithm for routing \emph{any} matching on any subset of $\Omega(n/\log n)$ of its vertices via edge-disjoint paths. We need a slightly different type of guarantee: the routing should be on {\bf vertex-disjoint} paths, and the graph degree may be super-constant (but still bounded). Rao and Zhou~\cite{RaoZhou} give such an algorithm, which is summarized in the next theorem. For completeness, we provide a proof sketch in Appendix.

\begin{theorem}[Theorem 7.1 in~\cite{RaoZhou}]\label{thm: vertex-disjoint routing on expanders}
Let $G=(V,E)$ be any $n$-vertex $d$-regular $\alpha$-expander, for $\alpha=1/2$.
Assume further that $n$ is even, and that the vertices of $G$ are partitioned into $n/2$ disjoint demand pairs $\mset=\set{(s_1,t_1),\ldots,(s_{n/2},t_{n/2})}$. Then there is an efficient algorithm that routes $\Omega\left (\frac n {\log n\cdot d^2}\right )$ of the demand pairs on vertex-disjoint paths in $G$.
\end{theorem}


\paragraph{The Cut-Matching Game}

We use the cut-matching game of Khandekar, Rao and Vazirani~\cite{KRV}. In this game, we are given a set $V$ of $N$ vertices, where $N$ is even, and two players: a cut player and a matching player. The goal of the cut player is to construct an expander $X$ on the set $V$ of vertices as quickly as possible, and the goal of the matching player is to delay its construction. The game is played in iterations. We start with the graph $X$ containing the set $V$ of vertices, and no edges.
In each iteration $j$, the cut player computes a bi-partition $(A_j,B_j)$ of the vertices of $V$ into two equal-sized sets, and the matching player returns some perfect matching $M_j$ between the two sets. The edges of $M_j$ are then added to $X$. The following theorem was proved in~\cite{KRV}.

\begin{theorem}[\cite{KRV}]\label{thm: KRV}
There is a probabilistic algorithm for the cut player, such that, no matter how the matching player plays, after $\gkrv(N)=O(\log^2N)$ iterations, graph $X$ is a $\half$-expander w.h.p.
\end{theorem}

\paragraph{Well-Linked Decompositions}
Well-linked decompositions have been used extensively in algorithms for network routing, e.g. in~\cite{Raecke,ANF,CKS,RaoZhou,Andrews}. We define below the specific type of well-linkedness that our algorithm uses and give an algorithm for computing the corresponding well-linked decomposition.

\begin{definition} Given a graph $G$, a subset $S$ of its vertices, and a parameter $\alpha>0$, we say that $S$ is $\alpha$-well-linked, iff for any partition $(A,B)$ of $S$, if we denote by $T_A=\out(A)\cap \out(S)$, and by $T_B=\out(B)\cap \out(S)$, then
$|E(A,B)|\geq \alpha\cdot \min\set{|T_A|,|T_B|}$.
\end{definition}

We also need a more general notion of well-linkedness that we define below. Intuitively, this definition of well-linkedness handles subsets $S$ of vertices, where $|\out(S)|$ may be large, but we will only be interested in routing small amounts of flow through $S$.

\begin{definition}
Let $S$ be any subset of vertices of a graph $G$. For any integer $k>0$ and for any $0<\alpha<1$, we say that set $S$ is $(k,\alpha)$-well-linked iff for any pair $T_1,T_2\sse \out(S)$ of disjoint subsets of edges, with $|T_1|+|T_2|\leq k$, the value of the minimum cut separating $T_1$ from $T_2$ in $G[S]\cup\out(S)$ is at least $\alpha\cdot\min\set{|T_1|,|T_2|}$. (We say that a cut $(X,Y)$ of $S$ separates $T_1$ from $T_2$ iff $T_1\sse \out(X)$ and $T_2\sse \out(Y)$.) \end{definition}

Note that if $|\out(S)|\leq k$, then set $S$ is $(k,\alpha)$-well-linked iff it is $\alpha$-well-linked, that is, the two definitions of well-linkedness become equivalent.
Notice also that if $S$ is $(k,\alpha)$-well-linked, then for any subset $T\sse\out(S)$ of at most $k$ edges, any matching on $T$ can be fractionally routed in graph $G[S]\cup \out(S)$ with congestion at most $2\beta(k)/\alpha$. This is since we can set up an instance of the sparsest cut problem on graph $G[S]\cup\out(S)$, where the edges of $T$ serve as terminals. Since $S$ is $(k,\alpha)$-well-linked, the value of the sparsest cut is at least $\alpha$, and so any matching on $T$ can be routed with congestion at most $2\beta(k)/\alpha$. 

Assume now that $S$ is not $(k,\alpha)$-well-linked. Then there must be a partition $(X,Y)$ of $S$, and two subsets $T_1\sse \out(X)\cap \out(S)$, $T_2\sse \out(Y)\cap \out(S)$ with $|T_1|+|T_2|\leq k$, such that $|E(X,Y)|<\alpha\cdot\min\set{|T_1|,|T_2|}$. We say that $(X,Y)$ is a \emph{$(k,\alpha)$-violating cut} for $S$.

Given a subset $S$ of vertices of $G$, we would like to find a partition $\wset$ of $S$, such that each set in $W\in\wset$ is $(k,\alpha)$-well-linked. We could do so using the standard well-linked decomposition procedures, for example like those used in~\cite{Raecke, CKS}. However,  in order to do so, we need to be able to check whether a given subset $W$ of vertices is $(k,\alpha)$-well-linked, and if not, find a $(k,\alpha)$-violating cut efficiently. We do not know how to do this, even approximately. Therefore, we will assume for now that we are given an oracle that finds a $(k,\alpha)$-violating cut in a given subset of vertices, if such cut exists. We describe the decomposition procedure and bound the number of edges $\sum_{W\in \wset}|\out(W)|$ in the resulting decomposition. When we use this decomposition later in the algorithm, we will be interested in routing small amounts of flow (up to $k$) across the clusters of the decomposition. Whenever we will be unable to route this flow, we will naturally obtain a $(k,\alpha)$-violating cut. Therefore, our algorithm itself will serve as an oracle to the decomposition procedure. We note that in the eventual decomposition $\wset$, not all sets $W\in \wset$ may be $(k,\alpha)$-well-linked, but we will be able to route the flow that we need to route across these clusters, and this is sufficient for us. We now describe the oracle-based decomposition procedure and analyze it.

We are given as input a subset $S$ of vertices of $G$, an integer $k$, and a parameter $0<\alpha<1$. Throughout the decomposition procedure, we maintain a partition $\wset$ of $S$, and at the beginning, $\wset=\set{S}$. The algorithm proceeds as follows. As long as not all sets in $\wset$ are $(k,\alpha)$-well-linked, our oracle computes a $(k,\alpha)$-violating partition $(X,Y)$ of one of the sets $W\in \wset$. We then remove $W$ from $\wset$ and add $X$ and $Y$ to $\wset$ instead. The next theorem bounds $\sum_{W\in \wset}|\out(W)|$.

\begin{theorem}\label{thm: oracle-well-linked}
Let $k>8$, and denote $\gamma=\gkrv(k)=\Theta(\log^2k)$. Let $\alpha(k)=\frac{1}{2^{11}\cdot \gamma\cdot\log k}$, and let $\wset$ be any partition of $S$ produced over the course of the above algorithm. Then $\sum_{W\in \wset}|\out(W)|\leq |\out(S)|\left (1+\frac 1 {64\gamma}\right )$.
\end{theorem}

We emphasize that the bound on $\sum_{W\in\wset}|\out(W)|$ holds for any partition produced over the course of the algorithm, and not just the final partition.

\begin{proof}
The proof uses a standard charging scheme. For simplicity, we denote $\alpha=\alpha(k)$. Consider some iteration of the algorithm, and suppose the oracle has found a $(k,\alpha)$-violating partition $(X,Y)$ of some set $W$ in the current partition. Let $T_X=\out(X)\cap \out(W)$, $T_Y=\out(Y)\cap \out(W)$, and assume w.l.o.g. that $|T_X|\leq |T_Y|$ (note that it is possible that $|T_X|>k$). We charge the edges of $T_X$ evenly for the edges in $E(X,Y)$. Specifically, if $|T_X|\geq k/2$, then $|E(X,Y)|\leq \alpha k/2$ must hold, and the charge to each edge in $T_X$ is at most $\frac{\alpha k}{2|T_X|}\leq \frac{\alpha k}{|\out(X)|}$. Otherwise, $|E(X,Y)|\leq \alpha \cdot |T_X|$, and the charge to each edge of $T_X$ is at most $\alpha$. In any case, $|\out(X)|=|T_X|+|E(X,Y)|<2|\out(W)|/3$, and $|\out(Y)|\leq |\out(W)|$.

Consider some edge $e=(u,v)\in \bigcup_{W\in \wset}\out(W)$. We analyze the charge to edge $e$. We first bound the charge via the vertex $u$. Let $i_1\leq i_2\leq \cdots\leq i_{\ell}$ be the iterations of the decomposition procedure in which $e$ was charged via vertex $u$, and for each $1\leq j\leq \ell$, let $z_j=|\out(W)|$, where $W$ is the cluster to which $u$ belonged at the end of iteration $i_j$. Note that for each $1<j\leq \ell$, $z_j<2z_{j-1}/3$. Let $j^*$ be the largest index for which $z_{j^*}>k/2$.
 Then the total charge to $e$ via $u$ in iterations $i_1,\ldots,i_{j^*}$ is at most:
\[\frac{\alpha k}{z_1}+\frac{\alpha k}{z_2}+\cdots +\frac{\alpha k}{z_{j^*}}\leq \frac{\alpha k}{z_{j^*}}\left (1+(2/3)+(2/3)^2+\cdots+(2/3)^{j^*-1}\right )<\frac{3\alpha k}{z_{j^*}}\leq 6\alpha\] 
 
  In each subsequent iteration, the charge to edge $e$ was at most $\alpha$, and the number of such iterations is bounded by $2\log k$. So the charge to edge $e$ via vertex $u$ is at  most $6\alpha+2\alpha\log k<4\alpha\log k$, and the total charge to edge $e$ is at most $8\alpha \log k\leq \frac{1}{2^8\gamma}$. 
  This however only accounts for the \emph{direct} charge. For example, some edge $e'\not \in \out(S)$, that was first charged to the edges in $\out(S)$, can in turn be charged for some other edges. We call such charging \emph{indirect}. If we sum up the indirect charge for every edge $e\in \out(S)$, we obtain a geometric series, and so the total direct and indirect amount charged to every edge $e\in \out(S)$ is at most $\frac{1}{128\gamma}$. We conclude that $\sum_{W\in \wset}|\out(W)|\leq |S|\left (1+\frac{1}{64\gamma}\right )$. (The additional factor of $2$ is due to the fact that each edge of the partition is counted twice in $\sum_{W\in \wset}|\out(W)|$ - once for each its endpoint).
\end{proof}

Let $\alphaWL(k)=\alpha(k)/\alphasc(k)=\Omega(1/(\log^{3.5}k))$. If $|\out(S)|\leq k$, then we can obtain a $(k,\alphaWL(k))$-well-linked decomposition of $S$ efficiently, by using the algorithm $\algsc$ for the Sparsest Cut problem as our oracle: In each iteration, for each $W\in \wset$, we apply the algorithm $\algsc$ to the corresponding instance of the sparsest cut problem (where the edges of $\out(W)$ are viewed as terminals). If the algorithm $\algsc$ returns a $(k,\alpha(k))$-violating cut $(X,Y)$ for any set $W\in \wset$, then we can proceed with the decomposition procedure as before. Otherwise, we are guaranteed that each set $W\in \wset$ is $\alphaWL(k)$-well-linked. We therefore have the following corollary.

\begin{corollary}\label{corollary: weak well linked}
Let $S$ be any subset of vertices of $G$, such that $|\out(S)|\leq k$. Then we can efficiently find a partition $\wset$ of $S$, such that for each $W\in \wset$, $|\out(W)|\leq k$, and it is $\alphawl(k)=\frac{1}{2^{11}\cdot \alphasc(k)\cdot\gkrv(k)\cdot \log k}=\Omega(1/(\log^{3.5}k))$-well-linked. Moreover, $\sum_{W\in \wset}|\out(W)|\leq |\out(S)|\left (1+\frac 1 {64\gkrv(k)}\right )$.
\end{corollary}

This finishes the description of the well-linked decomposition procedure. 
Throughout the paper, we use $\alpha(k)=\frac{1}{2^{11}\cdot \gkrv(k)\cdot\log k}$ to denote the parameter from Theorem~\ref{thm: oracle-well-linked}, and $\alphaWL(k)=\alpha(k)/\alphasc(k)$ the parameter from Corollary~\ref{corollary: weak well linked}.

\paragraph{The Grouping Technique}
The grouping technique was first introduced by Chekuri, Khanna and Shepherd~\cite{ANF}, and has since been widely used in algorithms for network routing~\cite{CKS, RaoZhou, Andrews}, as a means of boosting network connectivity and well-linkedness parameters. 
We summarize it in the following theorem.

\begin{theorem}\label{thm: grouping}
Suppose we are given a graph $G=(V,E)$, with weights $w(v)$ on vertices $v\in V$, and a parameter $p$. Assume further that for each $v\in V$, $0\leq w(v)\leq p$. Then we can find a partition $\gset$ of the vertices in $V$, and for each group $U\in \gset$, find a tree $T_U\sse G$, such that:

\begin{itemize}
\item For each $U\in \gset$, $p\leq w(U)\leq 3p$, where $w(U)=\sum_{v\in U}w(v)$.
\item For each $U\in \gset$, tree $T_U$ contains all vertices of $U$.
\item The trees $\set{T_U}_{U\in \gset}$ are edge-disjoint. 
\end{itemize}
\end{theorem}
\begin{proof}
Let $T$ be the spanning tree of the graph $G$, and assume that it is rooted at some vertex $r$. We perform a number of iterations, where in each iteration we delete some edges and vertices from $T$. For each vertex $v$ of the tree $T$, let $T_v$ denote the sub-tree rooted at $v$, and let $w(T_v)$ denote the total weight of all vertices in $T_v$. We build the partition $\gset$ of $V$ gradually. At the beginning, $\gset=\emptyset$. While $w(T)>3p$, we perform the following iteration:

\begin{itemize}
\item Let $v$ be the lowest vertex in the tree $T$, such that $w(T_v)>p$.
\item If $w(T_v)\leq 2p$, then we add a new group $U$ to $\gset$, containing all vertices of $T_v$, and we delete $T_v$ from the tree $T$, setting $T_U=T_v$.
\item Otherwise, let $u_1,\ldots,u_k$ be the children of $v$, and let $j$ be the smallest index, such that $\sum_{i=1}^jw(T_{u_i})\geq p$. We add a new group $U$ to $\gset$, consisting of all vertices in trees $T_{u_1},\ldots,T_{u_j}$. Notice that $w(U)\leq 2p$ must hold. We let $T_U$ be the sub-tree of $T$ consisting of $v$ and the trees $T_{u_1},\ldots,T_{u_j}$. We delete the trees $T_{u_1},\ldots,T_{u_j}$ from the tree $T$.
\end{itemize}

Notice that if, at the beginning of the current iteration, $w(T)>3p$, then at the end of the current iteration, $w(T)>p$ must hold. In the last iteration, when $w(T)\leq 3p$, we add a final group $U$ to $\gset$, containing all vertices currently in the tree $T$, and we let $T_U$ be the current tree $T$. It is easy to verify that all conditions of the theorem hold for the final partition $\gset$ of $V$.
\end{proof}

\paragraph{Remark.} We will sometimes use the grouping theorem in slightly different settings. The first such setting is when we are given a subset $\tset\sse V$ of vertices called terminals, and we would like to group them into groups of cardinality at least $p$ and at most $3p$. In this case we will think of all non-terminal vertices as having weight $0$, and terminal vertices as having weight $1$. Instead of finding a partition $\gset$ of all vertices, we will be looking for a partition $\gset'$ of the set $\tset$ of terminals. This partition is obtained from $\gset$ by ignoring the non-terminal vertices. Another setting in which we use the grouping theorem is when we are given a subset $E'\sse E$ of edges, and we would like to find a grouping $\gset$ of these edges into groups of at least $p$
and at most $3p$ edges. As before, we would also like to find, for each group $U\in \gset$, a tree $T_U$ containing all edges in $U$, and we require that the trees $\set{T_U}_{U\in \gset}$ are edge-disjoint. This setting can be reduced to the previous one, by sub-dividing each edge $e\in E'$ with a terminal vertex. It is easy to verify that Theorem~\ref{thm: grouping} can be applied in this setting as well.

\section{The Algorithm}
\subsection{The Starting Point}\label{subsec: preprocessing}
Our starting point is similar to that used in previous work on the problem~\cite{ANF,CKS,RaoZhou,Andrews}: namely, we use the standard multicommodity flow LP-relaxation for the EDP problem to partition our graph into several disjoint sub-graphs, that are well-linked for their respective sets of terminals, and solve the problem separately on each such sub-graph. Recall that the standard LP-relaxation for EDP is defined as follows.
For each $1\leq i\leq k$, we have an indicator variable $x_i$ for whether or not we route the pair $(s_i,t_i)$. Let $\pset_i$ denote the set of all paths connecting $s_i$ to $t_i$ in $G$. The LP relaxation is defined as follows.

\begin{eqnarray*}
\mbox(LP)&&\\
\max&\sum_{i=1}^kx_i&\\
\mbox{s.t.}&&\\
&\sum_{P\in \pset_i}f(P)\geq x_i&\forall 1\leq i\leq k\\
&\sum_{P: e\in P}f(P)\leq 1&\forall e\in E\\
&0\leq x_i\leq 1&\forall 1\leq i\leq k\\
&f(P)\geq 0&\forall 1\leq i\leq k,\forall P\in \pset_i\\
\end{eqnarray*}

While this LP has exponentially many variables, it can be efficiently solved using standard techniques, e.g. by using an equivalent polynomial-size LP formulation. Throughout the rest of the paper, we denote by $\opt$ the value of the optimal solution to the LP. Clearly, the value of the optimal solution to the \EDP problem instance is at most $\opt$.

We need the following definition.

\begin{definition}
Given a graph $G=(V,E)$, and a subset $\tset\sse V$ of vertices called terminals, we say that $\tset$ is \emph{flow-well-linked} in $G$, iff any matching $\mset$ on $\tset$ can be fractionally routed with congestion at most $2$ in $G$.
\end{definition}

The next theorem follows from previous work of Chekuri, Khanna and Shepherd~\cite{ANF,CKS}, and we provide a short proof sketch in Appendix for completeness.

\begin{theorem}\label{thm: starting point}
Suppose we are given a graph $G=(V,E)$ and a set $\mset$ of $k$ source-sink pairs in $G$. Then we can efficiently partition $G$ into a collection $G_1,\ldots,G_{\ell}$ of vertex-disjoint induced sub-graphs, and compute, for each $1\leq i\leq \ell$, a collection $\mset_i\sse \mset$ of source-sink pairs contained in $G_i$, such that:

\begin{itemize}
\item $\sum_{i=1}^{\ell}|\mset_i|=\Omega(\opt/\log^2k)$.

\item For each $1\leq i\leq \ell$, if $\tset_i$ denotes the set of terminals participating in pairs in $\mset_i$, then $\tset_i\sse V(G_i)$, and moreover $G_i$ is flow-well-linked for $\tset_i$.
\end{itemize}
\end{theorem}

We now proceed to solve the problem on each one of the graphs $G_i$ separately. In order to simplify the notation, we denote the graph $G_i$ by $G$, the set $\mset_i$ of the source-sink pairs by $\mset$, and the set of terminals by $\tset$. For simplicity, we denote $|\mset|=k$. Recall that $G$ is flow-well-linked for $\tset$, the degree of every terminal in $\tset$ is $1$, and the maximum vertex degree in $G$ is at most $4$. It is now enough to prove that we can route $\Omega\left(\frac{k}{\log^{23.5}k\log\log k}\right )$ demand pairs in $\mset$ with congestion at most $14$.
We also assume that $k>k_0$, where $k_0$ is a large enough constant: otherwise, we can simply pick any source-sink pair $(s,t)\in \mset$, connect it with any path $P$ and output this as a solution. In particular, we will assume that $k>\log^{24}k$, and $\gkrv(k)=\Theta(\log^2k)>20$.

\subsection{Legal Contracted Graph}
Let $\gamma=\gkrv(k)=\Theta(\log^2k)$.
We use a parameter $k_1=\frac k{192\gamma^3\log \gamma}=\Omega\left (\frac k{\log^6k\log\log k}\right )$.

 We will maintain, throughout the algorithm, a graph $G'$, obtained from $G$ by contracting some subsets of non-terminal vertices of $G$. Specifically, we say that $G'$ is a \emph{legal contracted graph} for $G$, iff the following conditions holds:
 
 \begin{itemize}
 \item The set $V(G')$ is partitioned into two subsets, $V_1\sse V(G)$ containing the original vertices of $G$, and $V_2=V(G')\setminus V_1$, containing super-nodes $v_C$, for $C\sse V(G)$. The subsets $V_1$ and $\set{C}_{v_C\in V_2}$ of vertices of $G$ are all pairwise disjoint, and $\tset\sse V_1$.
 
 \item Graph $G'$ can be obtained from graph $G$ by contracting each cluster in set $\set{C\mid v_C\in V_2}$ into the super-node $v_C$ (we delete all self-loops, but we do not delete parallel edges). 

\item For each super-node $v_C\in V_2$,  $|\out_G(C)|\leq k_1$, and set $C$ is $\alphaWL(k)$-well-linked in graph $G$ for the edges in $\out_G(C)$. 
\end{itemize}

Notice that graph $G'$ may have parallel edges, and it remains flow-well-linked for the set $\tset$ of terminals. Also, since the maximum vertex degree in $G$ is constant, the maximum vertex degree in $G'$ is at most $k_1$, and each terminal has degree $1$. Each edge in graph $G'$ corresponds to some edge in the original graph $G$, and we will not distinguish between them. In particular, for every vertex subset $S'\sse V(G')$, if $S\sse V(G)$ is the corresponding subset of vertices in $G$, where every super-node $v_C\in S'\cap V_2$ is replaced by the vertices of $C$, then there is a one-to-one mapping between $\out_{G'}(S')$ and $\out_G(S)$, and we will identify the edges in these two sets, that is, $\out_{G'}(S')=\out_G(S)$.
We need the following simple claim.
\begin{claim}\label{claim: legal graph has many edges}
If $G'$ is a legal contracted graph for $G$, then $G'\setminus \tset$ contains at least $k/6$ edges.
\end{claim}
\begin{proof}
For each terminal $t\in \tset$, let $e_t$ be the unique edge adjacent to $t$ in $G'$, and let $u_t$ be the other endpoint of $e_t$. We partition the terminals in $\tset$ into groups, where two terminals $t,t'$ belong to the same group iff $u_t=u_{t'}$. Let $\gset$ be the resulting partition of the terminals. Since the degree of every vertex in $G'$ is at most $k_1$, each group $U\in \gset$ contains at most $k_1$ terminals. Next, we partition the terminals in $\tset$ into two subsets $X,Y$, where $|X|,|Y|\geq k/3$, and for each group $U\in \gset$, either $U\sse X$, or $U\sse Y$ holds. It is possible to find such a partition by greedily processing each group $U\in \gset$, and adding all terminals of $U$ to one of the subsets $X$ or $Y$, that currently contains fewer terminals. Finally, we remove terminals from set $X$ until $|X|=k/3$, and we do the same for $Y$. Since graph $G'$ is flow-well-linked for the terminals, it is possible to route $k/3$ flow units from the terminals in $X$ to the terminals in $Y$, with congestion at most $2$. Since no group $U$ is split between the two sets $X$ and $Y$, each flow-path must contain at least one edge of $G'\setminus \tset$. Therefore, the number of edges in $G'\setminus \tset$ is at least $k/6$.
\end{proof}

\subsection{Families of Good Vertex Subsets}
We define a good family of vertex subsets in graph $G$. We then proceed in two steps. First, we show that we can efficiently find a good family of vertex subsets in graph $G$. Next, we show that given such good family, we can find the desired routing of the source-sink pairs in $\mset$.

 \begin{definition}
 We say that a subset $S\sse V(G)\setminus \tset$ of vertices is a \emph{good subset} iff there is a subset $\Gamma\sse \out_G(S)$ of edges, with $|\Gamma|=k_1$, such that:
 
 \begin{itemize} 
 \item $S$ is $\alphawl(k)$-well-linked for $\Gamma$. That is, for any partition $(X,Y)$ of $S$, if $\Gamma_X=\Gamma\cap \out(X)$ and $\Gamma_Y=\Gamma\cap \out(Y)$, then $|E_G(X,Y)|\geq \alphawl (k)\cdot \min\set{|\Gamma_X|,|\Gamma_Y|}$.
 
 \item  There is a flow $F$ in graph $G$, where every edge $e\in \Gamma$ sends one flow unit to a distinct terminal $t_e\in\tset$  (so for $e\neq e'$, $t_e\neq t_{e'}$), and the congestion caused by $F$ is at most $2\beta(k)/\alphawl(k)=O(\log^{4.5}k)$.
\end{itemize}

We say that a family $\fset=\set{S_1,\ldots,S_\gamma}$ of $\gamma=\gkrv(k)=\Theta(\log^2 k)$ subsets of vertices is \emph{good} iff each subset $S_j$ is a good subset of vertices of $G$, and $S_1,\ldots,S_{\gamma}$ are pairwise disjoint.
 \end{definition}
 
We view the subset $\Gamma\sse \out_G(S)$ of edges as part of the definition of a good subset of vertices. In particular, when we say that we are given a good family $\fset=\set{S_1,\ldots,S_{\gamma}}$ of vertex subsets, we assume that we are also given the corresponding subsets $\Gamma_j\sse \out_G(S_j)$ of edges, for all $1\leq j\leq \gamma$.
 We use the next theorem to find a good family of vertex subsets in $G$.

 \begin{theorem}\label{thm: find a good family or a contractible set}
 Let $G'$ be a legal contracted graph for $G$. Then there is an efficient randomized algorithm that w.h.p. either returns a good family $\fset=\set{S_1,\ldots,S_{\gamma}}$ of vertex subsets in $G$, together with the corresponding subsets $\Gamma_j\sse \out_G(S_j)$ of edges for all $1\leq j\leq \gamma$, or finds a legal contracted graph $G''$ for $G$, with $|E(G'')|<|E(G')|$.
 \end{theorem}
 \begin{proof}
 Let $m$ be the number of edges in $G'\setminus \tset$. From Claim~\ref{claim: legal graph has many edges}, $m\geq k/6$.
 The proof consists of two steps. First, we randomly partition the vertices in $G'\setminus \tset$ into $\gamma$ subsets $X_1,\ldots,X_{\gamma}$. We show that with high probability, for each $1\leq j\leq \gamma$, $|\out_{G'}(X_j)|< \frac{10m}{\gamma}$, while  the number of edges with both endpoints in $X_j$, $|E_{G'}(X_j)|\geq \frac{m}{2\gamma^2}$. Therefore, $|E_{G'}(X_j)|>\frac{|\out_{G'}(X_j)|}{20\gamma}$ w.h.p. For each $j: 1\leq j\leq \gamma$, we then try to recover a good subset $S_j$ of vertices from the cluster $X_j$. If we succeed, then we obtain a good family $\fset=\set{S_1,\ldots,S_{\gamma}}$ of vertex subsets. If we fail to recover a good vertex subset for some $1\leq j\leq \gamma$, then we will produce a legal contracted graph $G''$ containing fewer edges than $G'$.
 
We start with the first part. We partition the vertices in $V(G')\setminus \tset$ into subsets $X_1,\ldots,X_{\gamma}$, where each vertex $v\in V(G')\setminus \tset$ selects an index $1\leq j\leq \gamma$ independently uniformly at random, and is then added to $X_j$. We need the following claim.

\begin{claim}\label{claim: random partition into gamma sets} With probability at least $\half$, for each $1\leq j\leq \gamma$, $|\out_{G'}(X_j)|< \frac{10m}{\gamma}$, while $|E_{G'}(X_j)|\geq \frac{m}{2\gamma^2}$.
\end{claim}
 
\begin{proof}
Let $H=G'\setminus\tset$.
Fix some $1\leq j\leq \gamma$. 
Let $\event_1(j)$ be the bad event that $\sum_{v\in X_j}d_{H}(v)\geq \frac{2m}{\gamma}\cdot \left (1+\frac 1 {\gamma}\right )$. In order to bound the probability of $\event_1(j)$, we define, for each vertex $v\in V(H)$, a random variable $x_v$, whose value is $\frac{d_{H}(v)}{k_1}$ if $v\in X_j$ and $0$ otherwise. Notice that $x_v\in [0,1]$, and the random variables $\set{x_v}_{v\in V(H)}$ are pairwise independent. Let $B=\sum_{v\in V(H)} x_v$. Then the expectation of $B$, $\mu_1=\sum_{v\in V(H)} \frac{d_{H}(v)}{\gamma k_1}=\frac{2m}{\gamma k_1}$. Using the standard Chernoff bound (see e.g. Theorem 1.1 in~\cite{measure-concentration}),

\[\prob{\event_1(j)}=\prob{B> \left (1+1/\gamma\right )\mu_1}\leq e^{-\mu_1/(3\gamma^2)}=e^{-\frac{2m}{3\gamma^3 k_1}}<\frac 1 {6\gamma}\]

since $m\geq k/6$ and $k_1=\frac{k}{192\gamma^3\log \gamma}$.

For each terminal $t\in \tset$, let $e_t$ be the unique edge adjacent to $t$ in graph $G'$, and let $u_t$ be its other endpoint. Let $U=\set{u_t\mid t\in \tset}$. For each vertex $u\in U$, let $w(u)$ be the number of terminals $t$, such that $u=u_t$. Notice that $w(u)\leq k_1$ must hold. We say that a bad event $\event_2(j)$ happens iff $\sum_{u\in U\cap X_j}w(u)\geq \frac k{\gamma}\cdot \left (1+\frac 1 {\gamma}\right )$. 
In order to bound the probability of the event $\event_2(j)$, we define, for each $u\in U$, a random variable $y_u$, whose value is $w(u)/k_1$ iff $u\in X_j$, and it is $0$ otherwise. Notice that $y_u\in [0,1]$, and the variables $y_u$ are independent for all $u\in U$. Let $Y=\sum_{u\in U} y_u$. The expectation of $Y$ is $\mu_2=\frac{k}{k_1\gamma}$, and event $\event_2(j)$ holds iff $Y\geq \frac{k}{k_1\gamma}\cdot  \left (1+\frac 1 {\gamma}\right )\geq \mu_2\cdot \left (1+\frac 1 {\gamma}\right )$. Using the standard Chernoff bound again, we get that:

\[\prob{\event_2(j)}\leq e^{-\mu_2/(3\gamma^2)}\leq e^{-k/(3k_1\gamma^3)}\leq \frac 1 {6\gamma}\]

since $k_1=\frac{k}{192\gamma^3\log \gamma}$. Notice that if events $\event_1(j),\event_2(j)$ do not hold, then:

\[|\out_{G'}(X_j)|\leq \sum_{v\in X_j}d_{H}(v)+\sum_{u\in U\cap X_j}w(u)\leq \left (1+\frac 1 {\gamma}\right )\left (\frac{2m}{\gamma}+\frac{k}{\gamma}\right )< \frac{10m}{\gamma}\]

since $m\geq k/6$.

Let $\event_3(j)$ be the bad event that $|E_{G'}(X_j)|< \frac{m}{2\gamma^2}$. We next prove that $\prob{\event_3(j)}\leq \frac 1 {6\gamma}$. We say that two edges $e,e'\in E(G'\setminus \tset)$ are \emph{independent} iff they do not share any endpoints. Our first step is to compute a partition $U_1,\ldots,U_r$ of the set $E(G'\setminus \tset)$ of edges, where $r\leq 2k_1$, such that for each $1\leq i\leq r$, $|U_i|\geq \frac m{4k_1}$, and all edges in set $U_i$ are mutually independent. In order to compute such a partition, we construct an auxiliary graph $Z$, whose vertex set is $\set{v_e\mid e\in E(H)}$, and there is an edge $(v_e,v_{e'})$ iff $e$ and $e'$ are not independent. Since the maximum vertex degree in $G'$ is at most $k_1$, the maximum vertex degree in $Z$ is bounded by $2k_1-2$. Using the Hajnal-Szemer\'edi Theorem~\cite{Hajnal-Szemeredi}, we can find a partition $V_1,\ldots,V_r$ of the vertices of $Z$ into $r\leq 2k_1$ subsets, where each subset $V_i$ is an independent set, and $|V_i|\geq \frac{|V(Z)|}{r}-1\geq \frac{m}{4k_1}$. The partition $V_1,\ldots,V_r$ of the vertices of $Z$ gives the desired partition $U_1,\ldots,U_r$ of the edges of $G'\setminus \tset$. For each $1\leq i\leq r$, we say that the bad event $\event_3^i(j)$ happens iff $|U_i\cap E(X_j)|< \frac{|U_i|}{2\gamma^2}$. Notice that if $\event_3(j)$ happens, then event $\event_3^i(j)$ must happen for some $1\leq i\leq r$. Fix some $1\leq i\leq r$. The expectation of $|U_i\cap E(X_j)|$ is $\mu_3=\frac{|U_i|}{\gamma^2}$. Since all edges in $U_i$ are independent, we can use the standard Chernoff bound to bound the probability of $\event_3^i(j)$, as follows:

\[\prob{\event_3^i(j)}=\prob{|U_i\cap E(X_j)|<\mu_3/2}\leq e^{-\mu_3/8}=e^{-\frac{|U_i|}{8\gamma^2}}\]

Since $|U_i|\geq \frac{m}{4k_1}$, $m\geq k/6$, $k_1=\frac{k}{192\gamma^3\log \gamma}$, and $\gamma=\Theta(\log^2k)$, this is bounded by $\frac{1}{12k_1\gamma}$. We conclude that $\prob{\event_3^i(j)}\leq \frac{1}{12k_1\gamma}$, and by using the union bound over all $1\leq i\leq r$, $\prob{\event_3(j)}\leq \frac{1}{6\gamma}$.

 Using the union bound over all $1\leq j\leq \gamma$, with probability at least $\half$, none of the events $\event_1(j),\event_2(j),\event_3(j)$ for $1\leq j\leq \gamma$ happen, and so for each $1\leq j\leq \gamma$, $|\out_{G'}(X_j)|< \frac{10m}{\gamma}$, and $|E_{G'}(X_j)|\geq\frac{m}{2\gamma^2}$ must hold.
\end{proof}
 

 Given a partition $X_1,\ldots,X_{\gamma}$, we can efficiently check whether the conditions of Claim~\ref{claim: random partition into gamma sets} hold. If they do not hold, we repeat the randomized partitioning procedure.  From Claim~\ref{claim: random partition into gamma sets}, we are guaranteed that w.h.p., after $\poly(n)$ iterations, we will obtain a partition with the desired properties. Assume now that we are given the partition $X_1,\ldots,X_{\gamma}$ of $V(G')\setminus \tset$, for which the conditions of Claim~\ref{claim: random partition into gamma sets} hold. Then for each $1\leq j\leq \gamma$, $|E_{G'}(X_j)|>\frac{|\out_{G'}(X_j)|}{20\gamma}$. Let $X'_j\sse V(G)\setminus \tset$ be the set obtained from $X_j$, after we un-contract each cluster, that is, for each super-node $v_C\in V_2\cap X_j$, we replace $v_C$ with the vertices of $C$. Notice that $\set{X'_j}_{j=1}^{\gamma}$ is a partition of $V(G)\setminus\tset$.
We now proceed as follows. For each $1\leq j\leq \gamma$, we perform a partitioning procedure for the set $X'_j$ of vertices. We say that this partitioning procedure is successful, iff we find a good subset $S_j\sse X'_j$ of vertices. Therefore, if the partitioning procedure is successful for all $j$, then we obtain a good family $(S_1,\ldots,S_{\gamma})$ of disjoint vertex subsets. If the partitioning procedure is not successful for some $j$, then we will produce a legal contracted graph $G''$ as required. 

We now describe the partitioning procedure for some $j: 1\leq j\leq \gamma$. 
Intuitively, we would like to perform a well-linked decomposition of the set $X'_j$ of vertices, using Theorem~\ref{thm: oracle-well-linked}, to obtain a partition $\wset_j$ of $X'_j$. If we could ensure that each set $W\in\wset_j$ has $|\out_G(W)|\leq k_1$, and it is $\alphaWL(k)$-well-linked, then we could simply obtain the graph $G''$ by first uncontracting all clusters $C$ with $v_C\in V_2\cap X_j$, and then contracting all clusters in $\wset_j$ into super-nodes. Since we are guaranteed that $\sum_{W\in \wset_j}|\out_G(W)|\leq |\out_G(X_j')|(1+\frac{1}{64\gamma})$, while $|E_{G'}(X_j)|>\frac{|\out_G(X'_j)|}{20\gamma}$, it is easy to verify that $|E(G'')|< |E(G')|$ would hold. There are two problems with this approach. First, in order to use Theorem~\ref{thm: oracle-well-linked}, we need an oracle for finding $(k,\alpha(k))$-violating cuts of sets. Second, even if we had such an oracle, we would not be able to guarantee that for each set $W\in \wset$, $|\out_G(W)|\leq k_1$.  On the other hand, if, for some set $W\in \wset_j$, $|\out_{G}(W)|\geq k_1$, then it is possible that $W$ is a good set, though this is not guaranteed. Our idea is to gradually perform the well-linked decomposition of the set $X_j'$, using Theorem~\ref{thm: oracle-well-linked}. 
We will maintain the current partition $\wset_j$ of $X_j'$ into clusters, and in addition, a partition of $\wset_j$ into two subsets: $\wset^1$ and $\wset^2$. Intuitively, $\wset^1$ contains all active clusters, that still participate in the well-linked decomposition procedure, and that we may still sub-divide into smaller clusters later, while $\wset^2$ contains inactive clusters.
In each iteration, we will select an arbitrary cluster $S\in \wset^1$, and check if $S$ is a good set of vertices. If so, then we declare the iteration successful, and stop the procedure. Otherwise, we will either obtain a $(k,\alpha(k))$-violating cut of some set $S'\in \wset_j$, or we will be able to perform a different well-linked decomposition step that will turn cluster $S$ into an inactive one. We now give a formal description of the partitioning procedure.

Throughout the partitioning procedure, we maintain a partition $\wset_j$ of the set $X_j'$ of vertices, where at the beginning $\wset_j=\set{X_j'}$. 
Set $\wset_j$ is in turn partitioned into two subsets: set $\wset^1$ of active clusters and set $\wset^2$ of inactive clusters. At the beginning, $\wset^1=\wset_j$, and $\wset^2=\emptyset$.
We also maintain a graph $\tilde{G}$, which is an ``almost legal'' contracted graph for $G$ in the following sense. The set $V(\tilde G)$ of vertices is partitioned into two subsets, $\tV_1=V(\tilde G)\cap V(G)$ and $\tV_2=V(\tilde G)\setminus \tV_1$, with $\tset\sse \tV_1$. Each vertex $v_C\in \tV_2$ is associated with a cluster $C\sse V(G)\setminus \tV_1$, and all subsets $\set{C}_{v_C\in \tV_2}$ of vertices are pairwise disjoint. As before, we can obtain $\tilde G$ from $G$, by contracting each cluster $C$ (where $v_C\in \tV_2$) into a super-node $v_C$, and deleting self-loops. For each cluster $S\in \wset^1$, there is a super-node $v_S\in \tV_2$. Let $V_2'=\set{v_S\mid S\in \wset^1}$ be the set of all such super-nodes. Then for each super-node $v_C\in \tV_2\setminus V_2'$, $|\out_G(C)|\leq k_1$, and $C$ is $\alphawl(k)$-well-linked for $\out_G(C)$ in graph $G$. In other words, graph $\tilde{G}$ is a legal contracted graph for $G$, except for the super-nodes $v_S$, where $S\in \wset^1$: for such nodes $v_S$, we are not guaranteed that $|\out_G(S)|\leq k_1$, or that $S$ is well-linked. However, if $\wset^1=\emptyset$, then $\tilde{G}$ is a legal contracted graph of $G$. We remark that for clusters $S\in \wset^2$, graph $\tilde G$ does not necessarily contain a super-node $v_S$, and it is possible that the vertices of $S$ are split among several super-nodes. We only maintain the set $\wset^2$ for accounting purposes. The initial graph $\tilde{G}$ is obtained from $G'$ as follows: we un-contract all super-nodes $v_C\in X_j$, and then contract all vertices of $X'_j$ into a single super-node $v_{X'_j}$. We set $\wset_j=\wset^1=\set{X'_j}$ and $\wset^2=\emptyset$.
While $\wset^1$ is non-empty, we select any cluster $S\in \wset^1$ and process it. At the end of this procedure, we will either declare that $S$ is a good set, or we will find a $(k,\alpha(k))$-violating cut of some cluster $S'\in \wset^1$, or $S$ will become inactive.

Let $S\in \wset^1$ be the current cluster. We try to send $k_1$ flow units from the edges of $\out_{\tilde{G}}(S)$ to  the terminals in $\tset$ in the current graph $\tilde{G}$  with no congestion. Two case are possible, depending on whether or not such flow exists.

\paragraph{Case 1:} Assume first that such flow exists. From the integrality of flow, there is a collection $\pset$ of $k_1$ edge-disjoint paths in $\tilde G$, each path connecting distinct edges in $\out_{\tilde{G}}(S)$ to distinct terminals in $\tset$. Let $\Gamma\sse \out_{\tilde{G}}(S)$ be the set of $k_1$ edges which serve as endpoints of paths in $\pset$. We set up an instance of the sparsest cut problem in graph $G[S]\cup \out_G(S)$, where the edges in set $\Gamma$ serve as terminals. We then run the algorithm \algsc on the resulting instance. If the algorithm returns a cut $(X,Y)$ of sparsity less than $\alpha(k)$, then $(X,Y)$ is a $(k,\alpha(k))$-violating cut for $S$. We then replace $S$ with $X$ and $Y$ in $\wset_j$ and in $\wset^1$. We also update the current graph $\tilde{G}$, by first un-contracting the super-node $v_S$, and then contracting the two clusters $X$ and $Y$ into super-nodes $v_X$ and $v_Y$, respectively. This ends the current iteration, and we then proceed to process some new set in $\wset^1$. Assume now that algorithm  \algsc returns a cut whose sparsity is at least $\alpha(k)$. Then we are guaranteed that  $S$ is $\alphawl(k)=\alpha(k)/\alphasc(k)$-well-linked for $\Gamma$. Recall that we are given a set $\pset$ of $k_1$ edge-disjoint paths connecting the edges in $\Gamma$ to the terminals $\tset$ in graph $\tilde G$, where each path connects a distinct edge $e\in \Gamma$ to a distinct terminal $t_e\in \tset$. In order for $S$ to be a good set, a low-congestion flow connecting the edges in $\Gamma$ to the terminals must exist in the original graph $G$. We will try to find this flow, as follows. The flow will follow the paths in $\pset$, except that we need to specify how the flow is routed inside each cluster $C$ for $v_C\in \tV_2$. Observe that for each such cluster $C$, the paths in $\pset$ define a set $D_C$ of $1$-restricted demands on $\out_G(C)$. Moreover, the total number of edges in $\out_G(C)$ participating in the paths in $\pset$ is at most $k_1$, as there are only $k_1$ paths in $\pset$ and we can assume w.l.o.g. that they are simple. If $v_C\not\in V_2'$, then we are guaranteed that graph $G[C]\cup \out_G(C)$ is $\alphawl(k)$-well-linked for $\out_G(C)$. Therefore, we can route the set $D_C$ demands inside $G[C]\cup\out_G(C)$ with congestion at most $2\beta(k)/\alphawl(k)$. If $v_C\in V'_2$, then $C\in\wset^1$, and it is possible that we cannot route the set $D_C$ of demands inside $G[C]\cup \out_G(C)$ with congestion at most $2\beta(k)/\alphawl(k)$. We then proceed as follows. If, for each super-node $v_C\in V_2'$, we can route the set $D_C$ of demands inside $G[C]\cup \out_G(C)$ with congestion at most $2\beta(k)/\alphawl(k)$, then $S$ is a good set, and the $j$th iteration is successful. Otherwise, let $v_C\in V_2'$ be any super-node, for which such flow does not exist. Consider the instance of the sparsest cut problem defined on the graph $G[C]\cup\out_G(C)$, where the edges of $\out_G(C)$ with non-zero demand serve as terminals (recall that there are at most $k_1$ such edges). Then the value of the sparsest cut in this instance is at most $\alphawl(k)$, and so by applying algorithm \algsc on this instance of sparsest cut, we will obtain a $(k,\alpha(k))$-violating cut $(X,Y)$ for set $C$. We then remove $C$ from $\wset^1$ and from $\wset_j$, and add $X$ and $Y$ to $\wset^1$ and $\wset_j$ instead. We also update $\tilde G$ by un-contracting the super-node $v_C$ and contracting the clusters $X$ and $Y$ into super-nodes $v_X$ and $v_Y$, respectively, and end the current iteration. To conclude, if it is possible to send $k_1$ flow units with no congestion in graph $\tilde G$
between $\out_{\tilde G}(v_S)$ and $\tset$, then either $S$ is a good set, or we  find a $(k,\alpha(k))$-violating cut $(X,Y)$ of some cluster $C\in \wset^1$ (where possibly $C=S$).

\paragraph{Case 2:} Assume now that such flow does not exist. Then there is a cut $(X,Y)$ in graph $\tilde G$, where $\tset\sse Y$, $v_S\in X$, and $|E(X,Y)|< k_1$. (If $|\out_{\tilde G}(S)|< k_1$, then we set $X=\set{v_S}$). 
Let $A\sse V(G)\setminus \tset$ be the subset of vertices obtained from $X$ after we un-contract every super-node $v_C\in X$. Then $|\out_G(A)|<k_1$. We perform a well-linked decomposition of $A$, using Corollary~\ref{corollary: weak well linked}, and we denote the resulting partition of $A$ by $\wset(A)$. Recall that each set $C\in \wset(A)$ is guaranteed to be $\alphawl(k)$-well-linked, and $|\out_G(C)|<k_1$. Moreover, $\sum_{C\in \wset(A)}|\out_G(C)|\leq |\out_G(A)|\left(1+\frac 1 {64\gamma}\right )\leq |\out_G(S)|\left(1+\frac 1 {64\gamma}\right )$. We say that the cluster $S\in \wset^1$ is responsible for $A$, and for the partition $\wset(A)$ (we will eventually charge the edges in $\out_G(S)$ for the edges in $\bigcup_{C\in \wset(A)}\out_G(C)$). 
We update the graph $\tilde G$, by first un-contracting all super-nodes that belong to $X$, and then contracting each cluster $C\in \wset(A)$ into a super-node $v_C$. Also, for each vertex $v_C\in \wset^1$, if $v_C\in X$, then we move $C$ from $\wset^1$ to $\wset^2$, where it becomes an inactive cluster (notice that super-node $v_C$ may not exist in the new graph anymore, as the vertices of $C$ may end up being partitioned into several clusters by the contraction procedure). Observe that the cluster $S$ that is responsible for $A$ has been moved from $\wset^1$ to $\wset^2$ in the current iteration, and hence it becomes an inactive cluster.

This finishes the description of the decomposition procedure for $X_j$, for $1\leq j\leq \gamma$. In order to analyze it, it is enough to show that if this procedure was not declared successful, then the final graph $G''$, obtained at the end of the procedure, when $\wset^1=\emptyset$, contains fewer edges than $G'$. (We note that from the above discussion it is clear that $G''$ must be a legal contracted graph for $G$.) We bound the number of edges in $G''$ in two steps. First, we bound the number of edges in $\sum_{C\in \wset^2}|\out_G(C)|$. Observe that $\wset^2$ defines a partition of the set $X_j'$ of vertices of $G$. Moreover, this partition was obtained by performing an oracle-based well-linked decomposition of $X_j'$. Therefore, from Theorem~\ref{thm: oracle-well-linked}, $\sum_{C\in \wset^2}|\out_G(C)|\leq |\out_G(X_j')|\left (1+\frac 1 {64\gamma}\right )$.

 Next, we bound the number of edges in $G''$, by charging them to the edges of $\bigcup_{C\in \wset^2} \out_G(C)$. 
 Let $A_1,A_2,\ldots,A_{\ell}$ be all sets of vertices $A$ that were decomposed in iterations where Case 2 happened, in the order in which they were processed. Observe that all vertices of $X_j'$ are contained in $\bigcup_{i=1}^{\ell}A_i$, as all clusters in $\wset^2$ are contained in $\bigcup_{i=1}^{\ell}A_i$ (but  the sets $A_i$ are not necessarily disjoint). The set of edges of $G''$ can be partitioned into two subsets: $E_1=\set{e=(u,v)\mid e\in E(G')\cap E(G''); u,v\not\in X_j}$, and set $E_2$ containing all remaining edges. It is easy to see that $E_2\sse \bigcup_{i=1}^{\ell}(\bigcup_{C\in\wset(A_i)}\out_G(C))$. 
Indeed, let $e=(u,v)\in E_2$. Let $u',v'$ be the endpoints of the corresponding edge in the original graph $G$. Two cases are possible. If both $u,v\not\in X_j'$, then the only way that edge $e$ was added to the graph $\tilde{G}$ is when either $u'$ or $v'$ belonged to some set $A_i$. Let $i^*$ be the largest index for which $\set{u',v'}\cap A_{i^*}\neq \emptyset$. Then $e\in  \bigcup_{C\in\wset(A_{i^*})}\out_G(C)$ must hold. Otherwise, if at least one of the vertices (say $v'$) belongs to $X_j'$, then, since every vertex in $X_j'$ belongs to some inactive cluster at the end of the algorithm, there is at least one index $i$ such that $v'\in A_i$. Let $i^*$ be the largest index for which $\set{u',v'}\cap A_{i^*}\neq \emptyset$. Then $e\in  \bigcup_{C\in\wset(A_{i^*})}\out_G(C)$ must hold. Therefore, $E_2\sse \bigcup_{i=1}^{\ell}(\bigcup_{C\in\wset(A_i)}\out_G(C))$.

 Recall that for each set $A_i$, for $1\leq i\leq \ell$, we have a distinct cluster $S_i\in \wset^2$ responsible for $A_i$, and $\sum_{C\in \wset(A_i)}|\out_G(C)|\leq |\out_G(S_i)|\left(1+\frac 1 {64\gamma}\right )$
Therefore, the total number of edges in graph $G''$ is bounded by:
 
 \[\begin{split}|E(G'')|& \leq 
 |E(G')|-|E_{G'}(X_j)|-|\out_{G'}(X_j)|+|E_2|\\
 &\leq |E(G')|-|\out_{G'}(X_j)|\left(1+\frac 1 {20\gamma}\right )+\sum_{C\in \wset_2}|\out_G(C)|\left (1+\frac 1 {64\gamma}\right )
 \\&\leq |E(G')|-|\out_{G'}(X_j)|\left(1+\frac 1 {20\gamma}\right )+|\out_{G'}(X_j)|\left (1+\frac 1 {64\gamma}\right )^2\\
 &<|E(G')|
 \end{split}
 \]
 \end{proof}
 
 We are now ready to describe the algorithm for finding a good family of vertex subsets in graph $G$. We start with the graph $G'=G$, which is trivially a legal contracted graph, and repeatedly apply Theorem~\ref{thm: find a good family or a contractible set} to it. Since the number of edges in any legal contracted graph is at least $k/6$ by Claim~\ref{claim: legal graph has many edges}, we are guaranteed that after at most $|E(G)|$ iterations, the algorithm will produce a good family of vertex subsets w.h.p. We summarize the result of this section in the next corollary.
 
 \begin{corollary}
 There is an efficient randomized algorithm that w.h.p. computes a good family of vertex subsets in graph $G$.
 \end{corollary}

 \subsection{Finding the Routing}\label{subsec: from good family to routing}
 In this section we assume that we are given a good family $\fset=\set{S_1,\ldots,S_{\gamma}}$ of vertex subsets of $G$. For each $1\leq j\leq \gamma$, we are also given a subset $\Gamma_j\sse \out_G(S_j)$ of edges, such that $S_j$ is $\alphawl(k)$-well-linked for $\Gamma_j$, and there is a flow $F_j: \Gamma_j\connect_{\eta} \tset$, where each edge $e\in \Gamma_j$ sends one flow unit to a distinct terminal $t_e$, and the total congestion due to $F_j$ is at most $\eta=2\beta(k)/\alphawl(k)$.
 
 In order to find the final routing, we build an expander on a subset of terminals and embed it into graph $G$. More precisely, we select an arbitrary subset $\mset'\sse \mset$ of $k'/2$ source-sink pairs, where $k'=k/\poly\log k$. Let $\tset'\sse \tset$ be the subset of terminals participating in pairs in $\mset'$, and assume that $\tset'=\set{t_1,\ldots,t_{k'}}$.
 We construct an expander $X$ on the set $\set{v_1,\ldots,v_{k'}}$ of vertices, which is then embedded into the graph $G$ as follows. 
For each $1\leq i\leq k'$, we define a connected component $C_i$ in graph $G$, that represents the vertex $v_i$ of the expander. For each edge $e=(v_i,v_j)\in E(X)$, we define a path $P_e$, connecting a vertex of $C_i$ to a vertex of $C_j$ in $G$. We will ensure that each edge of $G$ may only appear in a small constant number of components $C_i$, and a small constant number of paths $P_e$. We also ensure that for each $1\leq i\leq k'$, terminal $t_i\in C_i$. We will think about the expander vertex $v_i$ as representing the terminal $t_i$. The idea is that any {\bf vertex-disjoint} routing of the terminal pairs in the expander $X$ can now be translated into a low edge-congestion routing in the original graph $G$. 

We now turn to describe the construction of the expander $X$ and the connected components $C_1,\ldots,C_{k'}$ that we use to embed $X$ into $G$. The construction exploits the good family $\fset=\set{S_1,\ldots,S_{\gamma}}$ of vertex subsets.
For each $1\leq i\leq k'$, we construct a collection $T_1,\ldots,T_{k'}$ of trees in graph $G$. Each such tree $T_i$ contains, for each $1\leq j\leq \gamma$, an edge $e_{i,j}\in \Gamma_j$. For each $1\leq j\leq \gamma$, the edges $e_{1,j},e_{2,j},\ldots,e_{k',j}$ are all distinct, and we think of the edge $e_{i,j}$ as the copy of the vertex $v_i\in V(X)$ for the set $S_j$. In other words, each tree $T_i$ spans $\gamma$ copies of the vertex $v_i$, one copy $e_{i,j}$ for each set $S_j\in \fset$. We will ensure that each edge of graph $G$ only participates in a constant number of such trees.
Additionally, we build a set $\pset=\set{P_t\mid t\in \tset'}$ of paths, where path $P_t$ connects the 
terminal $t$ to a distinct tree $T_i$ (so if $t\neq t'$, then $t$ and $t'$ are connected to different trees), and the total congestion caused by paths in $\pset$ is at most $4$. We rename the terminals in $\tset'$, so that $t_i$ denotes the terminal that is connected to the tree $T_i$. The final connected component $C_i$ is simply the union of the tree $T_i$ and the path $P_{t_i}$.

In order to construct the expander $X$ on the set $\set{v_1,\ldots,v_{k'}}$ of vertices, we use the cut-matching game of \cite{KRV}, where we use the sub-graph $G[S_j]$ of $G$ to route the $j$th matching between the corresponding copies $e_{1,j},e_{2,j},\ldots,e_{k',j}$ of the vertices $v_1,\ldots,v_{k'}$, respectively. Recall that we are only guaranteed that sets $\set{S_j}_{j=1}^{\gamma}$ are $\alphawl(k)$-well-linked for the edges in $\Gamma_j$, and so in order to route these matchings, we may have to incur the congestion of $\Omega(1/\alphaWL(k))$, which we cannot afford. However, this problem is easy to overcome by performing a suitable grouping of the edges of $\Gamma_j$.

The rest of the algorithm proceeds in three steps. In the first step, we perform groupings of the edges in the subsets $\Gamma_j$ for $1\leq j\leq \gamma$. In the second step, we construct the trees $T_1,\ldots,T_{k'}$. In the third step, we finish the construction of the expander $X$ and its embedding into $G$, and produce the final routing of a subset of demand pairs in $\mset'$.
 
 \paragraph{Step 1: Groupings.}
 In this step we compute, for each $1\leq j\leq \gamma$, a grouping of the edges in $\Gamma_j$. We then establish some properties of these groupings. We use the following two parameters: $p=8\beta(k)/\alphawl(k)=O(\log^{4.5}k)$ is the grouping parameter for the sets $\Gamma_j$. The second parameter, $k'=\frac{1}{2\gamma^3}\cdot \lfloor \frac{k_1}{6p}\rfloor=\Omega\left (\frac k {\log^{16.5}k\log\log k}\right )$
is the number of the vertices in the expander $X$ that we will eventually construct. We assume w.l.o.g. that $k'$ is even; otherwise we decrease its value by $1$.

Fix some $1\leq j\leq \gamma$. 
Since $G[S_j]\cup \out_G(S_j)$ is a connected graph, we can find a spanning tree $T_j$ of this graph, and perform a grouping of the edges in $\Gamma_j$ along this tree into groups whose size is at least $p$ and at most $3p$. Let $\gset_j$ be the resulting collection of groups, and let $k^*=\lfloor \frac{k_1}{6p}\rfloor $. For each group $U\in \gset_j$, let $T_j(U)$ be the sub-tree of the tree $T_j$ spanning the edges of $U$. For each group $U\in \gset_j$, we select one arbitrary representative edge, and we let $\Gamma'_j$ denote this set of representative edges. For each $e\in \Gamma'_j$, we denote by $U_e$ the group to which $e$ belongs. Additionally, let $U'_e\sse U_e$ be an arbitrary subset of $p$ edges of $U_e$, including $e$ itself. Notice that $|\Gamma'_j|\geq k^*$ must hold. If $|\Gamma'_j|>k^*$, then we discard edges from $\Gamma'_j$ arbitrarily, until $|\Gamma'_j|=k^*$ holds. This finishes the description of the grouping. The next theorem establishes some properties of the resulting groupings that will be used later. 

\ 

\begin{theorem}\label{thm: properties of grouping}

\ 
\begin{itemize}

\item For each $1\leq j\leq \gamma$, for any pair $X,Y\sse \Gamma'_j$ of edge subsets, where $|X|=|Y|$, there is a collection $\pset(X,Y):X\sconnect_{2} Y$ of paths contained in $G[S_j]\cup \out_G(S_j)$, where each path connects a distinct edge of $X$ to a distinct edge of $Y$, and the paths cause congestion at most $2$.

\item For all $1\leq i,j\leq \gamma$, there is a set $\pset_{i,j}:\Gamma'_i\sconnect_2\Gamma'_j$ of $k^*$ paths in graph $G$. That is, each path connects a distinct edge of $\Gamma'_i$ to a distinct edge of $\Gamma'_j$, with total congestion at most $2$.

\item Let $\Gamma^*_1\sse \Gamma'_1$ be any subset of $k'$ edges, $\mset'\sse \mset$ any subset of $k'/2$ source-sink pairs, and $\tset'$ the subset of terminals participating in pairs in $\mset'$. Then there is a set $\pset:\tset'\sconnect_4\Gamma^*_1$ of paths in $G$, each path  connecting a distinct terminal of $\tset'$ to a distinct edge of $\Gamma^*_1$, with total congestion at most $4$.
\end{itemize}
\end{theorem}
\begin{proof}
In order to prove the first assertion, fix some $1\leq j\leq \gamma$. From the integrality of flow, it is enough to prove that there is a flow $F_j(X,Y)$ in $G[S_j]\cup \out_G(S_j)$, where each edge in $X$ sends one flow unit, each edge in $Y$ receives one flow unit, and the flow congestion is at most $2$. We start by defining two subsets $X',Y'\sse \Gamma_j$ of edges, as follows: $X'=\bigcup_{e\in X}U'_e$, and $Y'=\bigcup_{e\in Y}U'_e$. Observe that $|X'|=|Y'|=|X|\cdot p$. Since set $S_j$ is $\alphawl(k)$-well-linked for $\Gamma_j$, there is a flow $F_j(X',Y')$  in $G[S_j]\cup \out_G(S_j)$, where every edge in $X'$ sends one flow unit, every edge in $Y'$ receives one flow unit, and the congestion due to this flow is at most $1/\alphawl(k)$. We are now ready to define the flow $F_j(X,Y)$. Each edge $e\in X$ spreads one flow unit uniformly among the edges of $U'_e$ along the tree $T_j(U_e)$. Next, all this flow is sent along the flow-paths in $F_j(X',Y')$, where we scale this flow down by factor $p$. Finally, each edge $e\in Y$ collects all flow from edges in $U'_e$ along the tree $T_j(U_e)$. Since all trees $\set{T_U}_{U\in \gset_j}$ are disjoint, and since the congestion caused by $F_j(X',Y')$ is at most $1/\alphawl(k)<p$, the resulting flow $F_j(X,Y)$ causes congestion at most $2$.

We now turn to prove the second assertion.
From the integrality of flow, it is enough to prove that there is a flow $F_{i,j}: \Gamma'_i\connect_2\Gamma'_j$, where every edge in $\Gamma'_i$ sends one flow unit and every edge in $\Gamma'_j$ receives one flow unit.
As before, we construct two edge subsets, $X\sse \Gamma_j$ and $Y\sse \Gamma_i$, as follows: $X=\bigcup_{e\in \Gamma'_i}U'_e$, and $Y=\bigcup_{e\in \Gamma'_j}U'_e$. Notice that $|X|=|Y|=k^*\cdot p$. 

Recall that from the definition of good vertex subsets, we already have a flow $F_j$, where each edge $e\in \Gamma_j$ sends one flow unit to a distinct terminal in $\tset$, with total congestion at most $\eta=2\beta(k)/\alphawl(k)$. We discard all flow-paths except those originating at the edges of $X$. As a result, we obtain a flow $F^{*}_j$, where each edge $e\in X$ sends one flow unit to a distinct terminal $t_e\in \tset$, and $F^{*}_j$ causes congestion at most $\eta$ in $G$. Let $\tset_j$ be the subset of terminals that receive flow in $F^{*}_j$, $|\tset_j|=|X|$. Similarly, we can define a flow $F^{*}_i$, where each edge $e\in Y$ sends one flow unit to a distinct terminal $t_e\in \tset$, and $F^{*}_i$ causes congestion at most $\eta$ in $G$. Subset $\tset_i$ of terminals is defined similarly. Notice that $\tset_i$ and $\tset_j$ are not necessarily disjoint. But since the set $\tset$ of terminals is flow-well-linked in $G$, there is a flow $F:\tset_i\sconnect_2\tset_j$, where each terminal in $\tset_i$ sends one flow unit, each terminal in $\tset_j$ receives one flow unit, and the congestion is at most $2$. We concatenate the three flows, $F_i^{*},F,F^{*}_j$, to obtain a flow $F':X\connect Y$. In this flow, each edge in $X$ sends one flow unit, each edge in $Y$ receives one flow unit, and the total congestion is at most $2\eta+2$.

We are now ready to define the flow $F_{i,j}$. Each edge $e\in \Gamma'_i$ sends one flow unit along the tree $T_{i}(U_e)$, which is evenly split among the edges of $U'_e$. We then use the flow $F'$, scaled down by factor $p$, to route this flow to the edges of $Y$. Finally, each edge $e\in \Gamma'_j$ collects the flow that the edges of $U'_e$ receive, along the tree $T_{j}(U_e)$, so that after collecting all that flow, edge $e$ receives $1$ flow unit. In order to analyze the total congestion due to flow $F_{i,j}$, observe that all trees $\set{T_i(U)}_{U\in \gset_i}\cup\set{T_j(U)}_{U\in \gset_j}$
are edge-disjoint. So the routing along these trees causes a congestion of at most $1$. Since flow $F'$ causes congestion of at most $2\eta+2$, and $p$ is selected so that $p\geq 2\eta+2$, the congestion due to the scaled-down flow $F'$ is at most $1$. The total congestion is therefore at most $2$.

Finally, we prove the third assertion. Let $\Gamma^*_1\sse \Gamma_1'$ be any subset of $k'$ edges, $\mset'\sse \mset$ any subset of $k'/2$ source-sink pairs, and $\tset'$ the set of all terminals participating in the pairs in $\mset'$. 
Let $X=\bigcup_{e\in \Gamma^*_1}U'_e$, so $|X|=k'p$. As before, 
we make use of the previously defined flow $F_1$, where each edge $e\in \Gamma_1$ sends one flow unit to a distinct terminal in $\tset$, with total congestion at most $\eta=2\beta(k)/\alphawl(k)$. We discard all flow-paths except those that originate at the edges of $X$. As a result, we obtain a flow $F^{*}$, where each edge $e\in X$ sends one flow unit to a distinct terminal $t_e\in \tset$, and $F^{*}$ causes congestion at most $\eta<p$ in $G$. 
We now define a new flow $F^{**}:\Gamma^*_1\connect_2 \tset$, where each edge in $\Gamma^*_1$ sends one flow unit, and each terminal in $\tset$ receives at most one flow unit. Flow $F^{**}$ is defined as follows. Each edge $e\in \Gamma^*_1$ sends one flow unit to the edges in set $U'_e$ along the tree $T_1(U_e)$, distributing it evenly among these edges. Each edge in $U'_e$ then sends the $1/p$ flow unit it receives from $e$ to the terminals via the flow $F^{*}$, so the flow $F^*$ is scaled down by factor $p$. Since the congestion caused by flow $F^*$ is $\eta<p$, and the trees $\set{T_1(U_e)}_{e\in \Gamma^*_1}$ are edge-disjoint, the total congestion caused by $F^{**}$ is at most $2$. Moreover, each terminal receives at most one flow unit in $F^{**}$. From the integrality of flow, there is a subset $\tset''\sse \tset$ of $k'$ terminals, and a collection $\pset_1:\Gamma^*_1\sconnect_2\tset''$ of paths in $G$. Since the set $\tset$ of terminals is flow-well-linked, using the integrality of flow, there is a collection $\pset_2: \tset''\sconnect_2\tset'$ of paths in $G$. We then obtain the desired collection $\pset$ of paths by concatenating the paths in $\pset_1$ with the paths in $\pset_2$.\end{proof}

\paragraph{Step 2: Constructing the Trees.}
The goal of this step is to find a collection $T_1,\ldots, T_{k'}$ of trees in graph $G$, such that each edge of $G$ belongs to at most $8$ trees. For each tree $T_i$, we will find a subset $E_i\sse E(T_i)$ of \emph{special edges}, that contains, for each $1\leq j\leq \gamma$, one edge $e_{i,j}\in \Gamma'_j$, such that the sets $E_1,\ldots,E_{k'}$ are pairwise disjoint. Notice that an edge $e\in \Gamma'_j$ may belong to several trees, but only to one of them as a special edge. 
For each $1\leq j\leq \gamma$, we denote $\Gamma^*_j=\set{e_{1,j},\ldots,e_{k',j}}$, the subset of edges of $\Gamma'_j$ that the trees $T_1,\ldots,T_{k'}$ contain as special edges. We summarize Step 2 in the next theorem.

\begin{theorem}\label{thm: building the trees}
Given a good family $\fset$, and a subset $\Gamma'_j\sse \out_G(S_j)$ of edges for each $1\leq j\leq \gamma$, as computed in Step 1, we can efficiently find $k'$ trees $T_1,\ldots,T_{k'}$ in graph $G$, and for each tree $T_i$ a subset $E_i\sse E(T_i)$ of special edges, such that:

\begin{itemize}
\item Each edge of $G$ belongs to at most $8$ trees;

\item Subsets $E_1,\ldots,E_{k'}$ of edges are pairwise disjoint; and

\item For all $1\leq i\leq k'$, $E_i=\set{e_{i,1},\ldots,e_{i,\gamma}}$, where  for all $1\leq j\leq \gamma$, $e_{i,j}\in \Gamma'_j$.
\end{itemize}
\end{theorem}

\begin{proof}
In order to prove the theorem, we start by augmenting the graph $G$ as follows. First, replace each edge of $G$ with two parallel edges. Next, for each $1\leq j\leq \gamma$, add a new vertex $s_j$, and for each edge $e\in \Gamma'_j$, we sub-divide one of the copies of $e$, by adding a new vertex $v_e$, which is then connected to the vertex $s_j$. 
Notice that from Theorem~\ref{thm: properties of grouping}, for each $1\leq j\neq j'\leq \gamma$, there are exactly $k^*$ edge-disjoint paths connecting $s_j$ to $s_{j'}$ in the resulting graph.
Finally, we replace each edge in the resulting graph by two bi-directed edges, thus obtaining a directed Eulerian graph that we denote by $G^+$. From Theorem~\ref{thm: properties of grouping}, for each pair $1\leq j\neq j'\leq \gamma$ of indices, there are $k^*$ edge-disjoint paths connecting $s_j$ to $s_{j'}$, and $k^*$ edge-disjoint paths connecting $s_{j'}$ to $s_j$. Notice also that each vertex $s_j$ has exactly $k^*$ incoming edges and exactly $k^*$ outgoing edges. 

As a next step, we use the standard edge splitting procedure in graph $G^+$. Our goal is to eventually obtain a graph $\tH$ on the set $\set{s_1,\ldots,s_{\gamma}}$ of vertices, such that each pair $s_j,s_{j'}$ is $k^*$-edge connected, and each edge $e=(s_j,s_{j'})\in E(\tH)$ is associated with a path $P_e$ connecting $s_j$ to $s_{j'}$ in $G^+$, while all paths in $\set{P_e\mid e\in E(\tH)}$ are edge-disjoint in $G^+$. 

Let $D=(V,A)$ be any directed multigraph with no self-loops. For any pair $(v,v')\in V$ of vertices, their connectivity $\lambda(v,v';D)$ is the maximum number of edge-disjoint paths connecting $v$ to $v'$ in $D$. Given a pair $a=(u,v)$, $b=(v,w)$ of edges, a splitting-off procedure replaces the two edges $a,b$ by a single edge $(u,w)$. We denote by $D^{a,b}$ the resulting graph. We use the extension of Mader's theorem~\cite{Mader} to directed graphs, due to Frank~\cite{Frank} and Jackson~\cite{Jackson}. Following is a simplified version of Theorem 3 from~\cite{Jackson}:

\begin{theorem}\label{thm: splitting-off}
Let $D=(V,A)$ be an Eulerian digraph, $v\in V$ and $a=(v,u)\in A$.
Then there is an edge $b=(w,v)\in A$, such that for all $y,y'\in V\setminus\set{v}$:
$\lambda(y,y';D)=\lambda(y,y';D^{ab})$
\end{theorem}

We apply Theorem~\ref{thm: splitting-off} repeatedly to all vertices of $G^+$ except for the vertices in set $\set{s_1,\ldots,s_{\gamma}}$, until we obtain a directed graph $\tH$, whose vertex set is $\set{s_1,\ldots,s_{\gamma}}$, and for each $1\leq j,j'\leq \gamma$, there are $k^*$ edge-disjoint paths connecting $s_j$ to $s_{j'}$ and $k^*$
edge-disjoint paths connecting $s_{j'}$ to $s_j$. Clearly, each edge $e=(s_j,s_{j'})\in E(\tH)$ is associated with a path $P_e$ connecting $s_j$ to $s_{j'}$ in $G^+$, and all paths $\set{P_e\mid e\in E(\tH)}$ are edge-disjoint.
Let $\tH'$ denote the undirected multi-graph identical to $\tH$, except that now all edges become undirected. 
Notice that each vertex $s_j$ must have $2k^*>>\gamma$ edges adjacent to it in $\tH'$, so the graph contains many parallel edges. For each pair $s_j,s_{j'}$ of vertices, there are exactly $2k^*$ edge-disjoint paths connecting $s_j$ to $s_{j'}$ in $\tH'$. For convenience, let us denote $2k^*$ by $\ell$.

As a next step, we build an auxiliary undirected graph $Z$ on the set $\set{s_1,\ldots,s_{\gamma}}$ of vertices, as follows. For each pair $s_j,s_{j'}$ of vertices, there is an edge $(s_j,s_{j'})$ in graph $Z$ iff there are at least $\ell/\gamma^3$ edges connecting $s_j$ and $s_{j'}$ in $\tH'$. If edge $e=(s_j,s_{j'})$ is present in graph $Z$, then its capacity $c(e)$ is set to be the number of edges connecting $s_j$ to $s_{j'}$ in $\tH'$. For each vertex $s_j$, let $C(s_j)$ denote the total capacity of edges incident on $s_j$ in graph $Z$.
We need the following simple observation.

\begin{observation}\label{observation: graph Z}
\ 
\begin{itemize}

\item For each vertex $v\in V(Z)$, $(1-1/\gamma^2)\ell\leq C(v)\leq \ell$. 
\item For each pair $(u,v)$ of vertices in graph $Z$, we can send at least $(1-1/\gamma)\ell$ flow units from $u$ to $v$ in $Z$ without violating the edge capacities.
\end{itemize}
\end{observation}
\begin{proof}
In order to prove the fist assertion, recall that each vertex in graph $\tH'$ has $\ell$ edges incident to it (this is since, in graph $G^+$, each vertex $s_1,\ldots,s_{\gamma}$ had exactly $k^*$ incoming and $k^*$ outgoing edges, and we did not perform edge splitting on these vertices). So $C(v)\leq \ell$ for all $v\in V(Z)$.
Call a pair $(s_j,s_{j'})$ of vertices bad iff there are fewer than $\ell/\gamma^3$ edges connecting $s_j$ to $s_{j'}$ in $\tH'$. Notice that each vertex $v\in V(Z)$ may participate in at most $\gamma$ bad pairs, as $|V(Z)|=\gamma$. Therefore, $C(v)\geq \ell-\gamma\ell/\gamma^3=\ell(1-1/\gamma^2)$ must hold.

For the second assertion, assume for contradiction that it is not true, and let $(u,v)$ be a violating pair of vertices. Then there is a cut $(A,B)$ in $Z$, with $u\in A$, $v\in B$, and the total capacity of edges crossing this cut is at most $(1-1/\gamma)\ell$. Since $u$ and $v$ were connected by $\ell$ edge-disjoint paths in graph $\tH'$, this means that there are at least $\ell/\gamma$ edges in graph $\tH'$ that connect bad pairs of vertices. But since we can only have at most $\gamma^2$ bad pairs, and each pair has less than $\ell/\gamma^3$ edges connecting them, this is impossible.
\end{proof}

We now proceed in two steps. First, we show that we can efficiently find a spanning tree of $Z$ with maximum vertex degree at most $3$. Next, using this spanning tree, we show how to construct the collection $T_1,\ldots,T_{k'}$ of trees. 

\begin{claim}\label{claim: small-degree spanning tree}
We can efficiently find a spanning tree $T^*$ of $Z$ with maximum vertex degree at most $3$.
\end{claim}
\begin{proof}
We use the algorithm of Singh and Lau~\cite{Singh-Lau} for constructing bounded-degree spanning trees. Suppose we are given a graph $G=(V,E)$, and our goal is to construct a spanning tree $T$ of $G$, where the degree of every vertex is bounded by $B$. For each subset $S\sse V$ of vertices, let $E(S)$ denote the subset of edges with both endpoints in $S$, and $\delta(S)$ the subset of edges with exactly one endpoint in $S$. Singh and Lau consider a natural LP-relaxation for the problem. We note that their algorithm works for a more general problem where edges are associated with costs, and the goal is to find a minimum-cost tree that respects the degree requirements; since we do not need to minimize the tree cost, we only discuss the unweighted version here. For each edge $e\in E$, we have a variable $x_e$ indicating whether $e$ is included in the solution. We are looking for a feasible solution to the following LP.

\begin{eqnarray}
&\sum_{e\in E}x_e=|V|-1&\label{LP: total edge weights}\\
&\sum_{e\in E(S)}x_e\leq |S|-1&\forall S\subset V \label{LP: sum for subsets}\\
&\sum_{e\in \delta(v)}x_e\leq B&\forall v\in V \label{LP: degree constraints}\\
&x_e\geq 0&\forall e\in E
\end{eqnarray}

Singh and Lau~\cite{Singh-Lau} show an efficient algorithm, that, given a feasible solution to the above LP, produces a spanning tree $T$, where for each vertex $v\in V$, the degree of $v$ is at most $B+1$ in $T$. Therefore, in order to prove the claim, it is enough to show a feasible solution to the LP, where $B=2$. Recall that $|V(Z)|=\gamma$. The solution is defined as follows. Let $e=(u,v)$ be any edge in $E(Z)$. We set the LP-value of $e$ to be $x_e=\frac{\gamma-1}{\gamma}\cdot \left (\frac{c(e)}{C(v)}+\frac{c(e)}{C(u)}\right )$. 
We say that $\frac{\gamma-1}{\gamma}\cdot \frac{c(e)}{C(v)}$ is the contribution of $v$ to $x_e$, and $\frac{\gamma-1}{\gamma}\cdot \frac{c(e)}{C(u)}$ is the contribution of $u$. We now verify that all constraints of the LP hold.

First, it is easy to see that $\sum_{e\in E}x_e=\gamma-1$, as required. Next, consider some subset $S\subset V$ of vertices. Notice that it is enough to establish Constraint~(\ref{LP: sum for subsets}) for subsets $S$ with $|S|\geq 2$. From Observation~\ref{observation: graph Z}, the total capacity of edges in $E_Z(S,\nots)$ must be at least $(1-1/\gamma)\ell$. Since for each $v\in S$, $C(v)\leq \ell$, the total contribution of the vertices in $S$ towards the LP-weights of edges in $E_Z(S,\nots)$ is at least $\frac{\gamma-1}{\gamma}\cdot (1-1/\gamma)=(1-1/\gamma)^2$. Therefore,

\[\sum_{e\in E(S)}x_e\leq \frac{\gamma-1}{\gamma}|S|-(1-1/\gamma)^2=|S|-|S|/\gamma-1-1/\gamma^2+2/\gamma\leq |S|-1\]

since we assume that $|S|\geq 2$. This establishes Constraint~(\ref{LP: sum for subsets}).
Finally, we show that for each $v\in V(Z)$, $\sum_{e\in \delta_v}x_e\leq 2$. First, the contribution of the vertex $v$ to  this summation is bounded by $1$. Next, recall that for each $u\in V(Z)$, $C(u)\geq (1-1/\gamma^2)\ell$, while the total capacity of edges in $\delta(v)$ is at most $\ell$. Therefore, the total contribution of other vertices to this summation is bounded by $\frac{\ell}{(1-1/\gamma^2)\ell}\cdot \frac{\gamma-1}{\gamma}\leq \frac{\gamma}{\gamma+1}\leq 1$. The algorithm of Singh and Lau can now be used to obtain a spanning tree $T^*$ for $Z$ with maximum vertex degree at most $3$.
\end{proof}

Root the tree $T^*$ at any degree-$1$ vertex $r$. Let $e=(s_i,s_j)$ be some edge of the tree, where $s_i$ is the parent of $s_j$. Recall that there are at least $\ell/\gamma^3$ edges $(s_i,s_j)$ in graph $\tH'$. Let $A(e)$ be any  collection of exactly $\ell/\gamma^3$ such edges. Recall that for each edge $e'\in A(e)$ in graph $\tH'$, there is a path $P$, connecting either $s_i$ to $s_j$ or $s_j$ to $s_i$ in graph $G^+$ (recall that graph $G^+$ is directed).
Since the direction of the edges in $G^+$ will not play any role in the following argument, we will assume w.l.o.g. that $P$ is directed from $s_j$ towards $s_i$. Recall that the first edge on path $P$ must connect $s_j$ to some vertex $v_{\te}$, where $\te\in \Gamma'_j$, and similarly, the last edge on path $P$ connects some vertex $v_{\te'}$, for $\te'\in \Gamma'_i$ to $s_i$. So by removing the first and the last edges from path $P$, we obtain a path $P_{e'}$ in graph $G$, that connects edge $\te\in \Gamma'_j$ to edge $\te'\in \Gamma'_i$. Since $s_i$ is the parent of $s_j$ in tree $T^*$, we will think of $P_{e'}$ as being directed from $S_j$ towards $S_i$. We call $\te$ \emph{the first edge of $P_{e'}$}, and $\te'$ \emph{the last edge of $P_{e'}$}. Going back to the edge $e=(s_i,s_j)$ in tree $T^*$, we can now define a set $\pset(e)=\set{P_{e'}\mid e'\in A(e)}$ of exactly $\ell/\gamma^3$ paths in graph $G$, associated with $e$. We let 

\[B_1(e)=\set{\te\in \Gamma'_j\mid \te \mbox{ is the first edge on some path $P_{e'}\in \pset(e)$}}\]

and

\[B_2(e)=\set{\te\in \Gamma'_i\mid \te \mbox{ is the last edge on some path $P_{e'}\in \pset(e)$}}\]

Both sets $B_1(e)$, $B_2(e)$ are multi-sets, that is, if some edge $\te\in \Gamma'_j$ appears as a first edge on two paths in $P_{e'}$, then we add two copies of $\te$ to $B_1(e)$. (From the construction of $G^+$, it is easy to see that $\te$ may appear as the first edge on at most two such paths). We then have that $\pset(e):B_1(e)\sconnect_4 B_2(e)$, since,
from the construction of graphs $G^+$ and $\tH'$, every edge of graph $G$ may appear on at most four paths of $\bigcup_{e\in E(T^*)}\pset(e)$.

We call the sets $B_1(e),B_2(e)$ of edges \emph{bundles corresponding to $e$}, and we view $B_1(e)$ as a bundle that belongs to $S_j$, while $B_2(e)$ is a bundle that belongs to $S_i$. Since the degree of tree $T^*$ is at most $3$, every set $S_j$ has at most three bundles that belong to it. From the construction of graph $G^+$, for every vertex $s_i: 1\leq i\leq \gamma$, each edge in $\Gamma'_i$ may appear at most twice in the multi-set defined by the union of the bundles that belong to $S_i$. In particular, it is possible that it appears twice in the same bundle. 
We need to make sure that this never happens. In order to achieve this, we will define, for each edge $e\in E(T^*)$, smaller bundles, $B_1'(e)\sse B_1(e)$ and $B_2'(e)\sse B_2(e)$, such that each edge appears at most once in each bundle, and there is a subset $\pset'(e)\sse \pset(e)$, where $\pset'(e):B_1'(e)\sconnect_4 B_2'(e)$. We will also ensure that $|B_1'(e)|=|B_2'(e)|=\frac{\ell}{4\gamma^3}$.

This is done as follows. Consider some edge $e=(s_i,s_j)$ in tree $T^*$, and assume that $s_i$ is the parent of $s_j$ in the tree. Consider first $B_1(e)$. For each edge $\te\in B_1(e)$, if two copies of $\te$ appear in $B_1(e)$, then we remove one of the copies from $B_1(e)$.  If $P\in\pset(e)$ is one of the two paths for which $\te$ is the first edge, then we remove $P$ from $\pset(e)$, and we also remove its last edge from $B_2(e)$. It is easy to see that we remove at most half the edges of $B_1(e)$. We then perform the same operation for $B_2(e)$. In the end, both $B_1(e)$ and $B_2(e)$ must contain at least a $1/4$ of the original edges, and $\pset(e)$ contains at least a $1/4$ of the original paths. We now let $\pset'(e)$ be any subset of exactly $\ell/4\gamma^3$ remaining paths, and we set $B_1'(e)$ to be the set of all edges $\te$ that appear as the first edge on some path in $\pset'(e)$, and similarly $B_2'(e)$ the set of all edges that appear as the last edge on some path in $\pset'(e)$. We perform this operation for all edges $e$ of tree $T^*$.

We are now ready to define the subsets $\Gamma^*_j\sse \Gamma'_j$ of $k'$ edges, $\Gamma^*_j=\set{e_{1,j},\ldots,e_{k',j}}$, that our trees will span. Fix some index $1\leq j\leq \gamma$. If $s_j$ is not the root of the tree $T^*$, then we let $\Gamma^*_j=B_1(e)$, where $e$ is the edge connecting $s_j$ to its father in $T^*$. If $s_j$ is the root of the tree, then $\Gamma^*_j=B_2(e)$, where $e$ is the unique edge incident on $s_j$ in tree $T^*$. Notice that $|\Gamma^*_j|=\frac{\ell}{4\gamma^3}=\frac{k^*}{2\gamma^3}=k'$.

Finally, we construct the trees $T_1,\ldots,T_k'$. In order to construct these trees, we process the vertices of the tree $T^*$ in the bottom-up order, starting from the leaves. Let $s_j$ be any vertex of $T^*$, and let $T^*(s_j)$ be the sub-tree of $T^*$, rooted at $s_j$. We will ensure that after vertex $s_j$ is processed, we will have a collection $T_1(s_j),\ldots,T_{k'}(s_j)$ of trees, such that for each vertex $s_i\in T^*(s_j)$, each one of the trees contains exactly one distinct edge of $\Gamma_i^*$ as a special edge. The trees $T_1(s_j),\ldots,T_{k'}(s_j)$ will consist of the union of the paths $\pset'(e)$, where $e$ is an edge in the sub-tree $T^*(s_j)$ of $T^*$, of the edges of $G$ whose both endpoints lie in sets $S_i$ for $s_i\in T^*(s_j)$, and of sets $\Gamma^*_i$, for $s_i\in T^*(s_j)$.

Assume first that $s_j$ is a leaf of $T^*$. Then the trees $T_1(s_j),\ldots,T_{k}(s_j)$ consist of a single distinct edge of $\Gamma^*_j$ each. Assume now that $s_j$ is an inner vertex of $T^*$. We will assume here that $s_j$ has two children, $s_a$ and $s_b$; the case where $s_j$ only has one child is treated similarly. 

Recall that we are given a collection $T_1(s_a),\ldots,T_{k'}(s_a)$ of trees spanning the sets $\Gamma_i^*$ of vertices $s_i$ in the sub-tree $T^*(s_a)$. We will assume w.l.o.g., that for each such tree $T_q(s_a)$, the root of the tree is an endpoint of the unique edge of $\Gamma^*_a$ that belongs to $T_q(s_a)$ as a special edge. Let $e=(s_a,s_j)$ be the edge of $T^*$ connecting $s_a$ to $s_j$. Recall that we are given a collection $\pset'(e):\Gamma^*_a\sconnect B_2(e)$ of paths in $G$. From Theorem~\ref{thm: properties of grouping}, we can find a set $\pset_1: B_2(e)\sconnect_2 \Gamma_j^*$ of paths contained in the sub-graph $G[S_j]$ of $G$, where each path in $\pset_1$ connects a distinct edge of $B_2(e)$ to a distinct edge of $\Gamma_j^*$. 
We now concatenate the paths in $\pset'(e)$ with the paths in $\pset_1$, to get a collection $\pset'_1$ of paths. Each path in $\pset'_1$ connects a root of a distinct tree $T_q(s_a)$ to a distinct edge of $\Gamma_j^*$.

Similarly, let $e'=(s_b,s_j)$ be the edge of $T^*$ connecting $s_b$ to $s_j$. We are again given a collection $\pset'(e'): \Gamma^*_b\sconnect B_2(e')$ of paths in $G$, and we can again find a set $\pset_2: B_2(e')\sconnect_2 \Gamma_j^*$ of paths contained in $G[S_j]$. Concatenating the paths in $\pset'(e')$ and $\pset_2$, we again obtain a collection $\pset_2'$ of paths, where each path connects a root of a distinct tree $T_q(s_b)$ with a distinct edge in $\Gamma_j^*$.

Consider now some edge $\te\in \Gamma_j^*$. We have two paths: $P_1\in \pset'_1$, connecting $\te$ to the root of some tree $T_q(s_a)$, and path $P_2\in \pset'_2$ connecting $\te$ to the root of some tree $T_{q'}(s_b)$. We obtain a tree $T_{\te}(s_j)$ by taking the union of $T_{q}(s_a),T_{q'}(s_b),P_1$ and $P_2$ (we may need to delete some edges to ensure that it is indeed a tree). The set of the special edges of this new tree consists of all special edges of $T_{q}(s_a),T_{q'}(s_b)$, and the edge $\te$.

At the end of this procedure, when the root $r$ of $T^*$ is processed, we will obtain a desired collection $T_1,\ldots,T_{k'}$ of trees, where for each $1\leq j\leq \gamma$, for each $1\leq i\leq k'$, tree $T_i$ contains an edge $e_{i,j}\in \Gamma^*_j$, and the edges $e_{1,j},\ldots,e_{k',j}$ are all distinct. We now analyze the congestion caused by these trees. First, as already observed, each edge of graph $G$ may belong to at most four paths of the set $\bigcup_{e\in E(T^*)}\pset'(e)$. Additionally, for each $1\leq j\leq \gamma$, we route two subsets of edges of $\Gamma'_j$ to each other twice. Each such routing causes congestion $2$ in graph $G[S_j]$, and so the total congestion caused by all these routings is at most $4$. We conclude that each edge of $G$ belongs to at most $8$ trees $T_1,\ldots,T_{k'}$.

\end{proof}

\paragraph{Step 3: Constructing the Expander and finding the routing.}

In this step, we construct the expander $X$, together with its embedding into the graph $G$, and find the final routing of a subset of demands in $\mset$. 
Let $\mset'\sse \mset$ be any subset of $k'/2$ demand pairs, and let $\tset'$ be the subset of terminals participating in the pairs of $\mset'$. 

Let $\pset=\tset'\sconnect_4\Gamma^*_1$ be the collection of paths connecting the terminals of $\tset'$ to the edges of $\Gamma^*_1\sse \Gamma_1$ (where $\Gamma_1^*=\set{e_{1,1},\ldots,e_{k',1}}$), guaranteed by Theorem~\ref{thm: properties of grouping}. Denote $\pset=\set{P_t\mid t\in \tset'}$, where $P_t$ is the path originating from terminal $t$. Rename the terminals in $\tset'$ as $\tset'=\set{t_1,\ldots,t_{k'}}$, where for each $1\leq i\leq k'$, $t_i$ is the terminal whose path $P_t$ terminates at the edge $e_{i,1}$ (the unique edge of $\Gamma^*_1$ that belongs to the tree $T_i$ as a special edge). For $1\leq i\leq k'$, let $C_i$ be the connected component of graph $G$, that consists of the union of the tree $T_i$ and the path $P_{t_i}$. Since each edge of graph $G$ participates in at most $8$ trees $T_i$, and at most $4$ paths in $\pset$, each edge of $G$ participates in at most $12$ connected components $C_i$. 

We now construct the expander $X$ and embed it into the graph $G$. The set of vertices of $X$ is $V(X)=\set{v_1,\ldots,v_{k'}}$, where we view each vertex $v_i$ as representing the terminal $t_i\in \tset'$. We view the connected component $C_i$ as the embedding of the vertex $v_i$ into $G$. Finally, we need to define the of the edges of $X$ and specify their embedding into $G$.
In order to do so, we use the cut-matching game of Khandekar, Rao and Vazirani~\cite{KRV} with $\gamma=\gkrv(k)$ iterations. Recall that in each iteration $j$, the cut player produces a partition $(A_j,B_j)$ of $V(X)$, with $|A_j|=|B_j|$. The matching player then returns some matching $M_j$ between the vertices of $A_j$ and $B_j$, and the edges of $M_j$ are added to graph $X$. We are guaranteed that no matter what the matching player does, there is always a way for the cut player to efficiently compute the partitions $(A_j,B_j)$ in each iteration $j$ (which may depend on the previous matchings $M_1,\ldots,M_{j-1}$), such that after $\gamma$ iterations, $X$ becomes a $\half$-expander w.h.p. 
Our idea is to use the graphs $G[S_j]$ to route the matchings $M_j$. Specifically, let $(A_1,B_1)$ be the partition of $V(X)$ produced by the cut player in the first iteration. Consider the set $\Gamma_1^*=\set{e_{1,1},\ldots,e_{k',1}}$ of edges. Partition $(A_1,B_1)$ of $V(X)$ defines a partition $(A_1',B_1')$ of these edges, where $A_1'=\set{e_{i,1}\mid v_i\in A_1}$ and $B_1'=\set{e_{i,1}\mid v_i\in B_1}$. From Theorem~\ref{thm: properties of grouping}, we can find a set $\qset_1:A_1'\sconnect_2 B_1'$ of $|A_1'|$ paths contained in $G[S_1]\cup \out_G(S_1)$, where each path in $\qset_1$ connects a distinct edge of $A_1'$ to a distinct edge of $B_1'$. Set $\qset_1$ of paths then defines a matching $M_1'$ between the sets $A_1'$ and $B_1'$, which in turn defines a matching $M_1$ between the sets $A_1$ and $B_1$ of vertices of $V(X)$. We then treat $M_1$ as the response of the matching player. For each edge $e=(v_i,v_{i'})\in M_1$ of the matching, we let $P_e$ be the unique path of $\qset_1$ connecting $e_{i,1}$ to $e_{i',1}$. We view $P_e$ as the embedding of $e$ into graph $G$. We continue similarly to execute the remaining iterations, where in each iteration $j: 1\leq j\leq \gamma$, we use the set $S_j\in\fset$ to find the matching $M_j$. That is, we define the partition $(A_j',B_j')$ of $\Gamma_j'$ based on the partition $(A_j,B_j)$ of $V(X)$ as before, find a collection $\qset_j: A_j'\sconnect_2 B_j'$ of paths contained in $G[S_j]\cup \out_G(S_j)$. These paths give us the matching $M_j'$ between the sets $A_j'$ and $B_j'$ of edges, which in turn gives us the matching $M_j$ between the sets $A_j$ and $B_j$ of vertices of $V(X)$. For each edge $e=(v_i,v_{i'})\in M_j$, we let $P_e$ be the unique path of $\qset_j$ connecting $e_{i,j}$ to $e_{i',j}$. We view $P_e$ as the embedding of $e$ into graph $G$. 
The final graph $X$ is the graph obtained after $\gamma$ iterations, with $E(X)=\bigcup_{j=1}^{\gamma}M_j$,
and we are guaranteed that w.h.p. it is a $\half$-expander. 
For each edge $e=(v_i,v_{i'})\in E(X)$, we have defined an embedding $P_e$ of $e$ into $G$, where $P_e$ is a path connecting some vertex in $C_i$ to some vertex in $C_{i'}$. Let $\pset_X=\set{P_e\mid e\in E(X)}$. Then $\pset_X=\bigcup_{j=1}^{\gamma}\qset_j$, and the total congestion caused by paths in $\pset_X$  in $G$ is at most $2$. This finishes the definition of the expander $X$ and of its embedding into $G$.

We now use the expander $X$ and its embedding into $G$, to route a subset of demand pairs. We identify from now on the vertices of $X$ with the terminals of $\tset'$ they represent, that is, $V(X)=\tset'$.

We use Theorem~\ref{thm: vertex-disjoint routing on expanders} to find a collection $\pset$ of $r=\Omega\left(\frac{k'}{\gamma^2\log k}\right )$ vertex-disjoint paths in the expander $X$, routing $r$ distinct demand pairs. Let $\mset''\sse \mset'$ be the set of these demand pairs, and assume w.l.o.g. that $\mset''=\set{(t_1,t_2),(t_3,t_4),\ldots,(t_{2r-1},t_{2r})}$. For each $1\leq i\leq r$,  let $P_i\in \pset$ be the path connecting $t_{2i-1}$ to $t_{2i}$.
In order to complete the routing, we transform each such path $P_i$ into a path $Q_i$ in graph $G$, connecting the same pair $(t_{2i-1},t_{2i})$ of terminals.

Fix some $1\leq i\leq r$. We now show how to transform the path $P_i$ connecting $t_{2i-1}$ to $t_{2i}$ in graph $X$ to a path $Q_i$ connecting the same pair of terminals in graph $G$.
In order to do so, we will replace the edges and the vertices of path $P_i$ by paths in graph $G$. First, each edge $e=(t_a,t_b)\in P_i$ is replaced by the path $P_e\sse G$, connecting some vertex $v\in C_a$ to some vertex $u\in C_b$. Next, consider some inner vertex $t_x\in P_i$, and let $e,e'$ be the two edges appearing immediately before and immediately after $t_x$ on the original path $P_i$, respectively. Let $v_x\in C_x$ be the last vertex on path $P_e$, and let $v'_x\in C_x$ be the first vertex on path $P_{e'}$. Then we replace the vertex $t_x$ with an arbitrary path $P_x$ connecting $v_x$ to $v'_x$ in the connected component $C_x$ of $G$. It now only remains to take care of the endpoints of path $P_i$. Let $e$ be the first edge on the original path $P_i$, and recall that the first vertex on $P_i$ is $t_{2i-1}$. Let $v_{2i-1}\in C_{2i-1}$ be the first vertex on the path $P_e$. Then we replace $t_{2i-1}$ by any path connecting $t_{2i-1}$ to $v_{2i-1}$ in the connected component $C_{2i-1}$. The last vertex of $P_i$ is taken care of similarly. Let $Q_i$ denote the resulting path. Notice that $Q_i$ consists of two types of segments: the first type are the paths $P_e$ for edges $e\in P_i$, and the second type is the paths $P_x$ for vertices $X\in {P_i}$. Let $Q_1,\ldots,Q_r$ be the resulting set of paths. We now bound the congestion due to paths in $Q_1,\ldots,Q_r$ in graph $G$. Recall that the paths $\set{P_i}_{i=1}^r$ are edge- and vertex-disjoint. Recall also that each edge of graph $G$ participates in at most $2$ paths of the set $\pset_X=\set{P_e\mid e\in E(X)}$. Therefore, the congestion due to type-1 segments in $\set{Q_i}_{i=1}^r$ is at most $2$. Since the paths in $\set{P_i}_{i=1}^r$ are vertex-disjoint, and every edge of graph $G$ participates in at most $12$ components $C_1,\ldots,C_{k'}$, the congestion due to type-2 segments is bounded by $12$. Overall, the paths in $\set{Q_i}_{i=1}^r$ cause congestion at most $14$. The number of demand pairs routed is $r=\Omega\left(\frac{k'}{\gamma^2\log k}\right )=\Omega\left (\frac{k}{\log^{21.5}k\log\log k}\right )$.


To conclude, we have started with a graph $G$, a collection $\mset$ of $k$ source-sink pairs, and the set $\tset$ of terminals participating in pairs in $\mset$, such that $G$ is flow-well-linked for $\tset$. We have constructed a routing for the subset $\mset''\sse \mset$ of $\Omega \left (\frac k {\log^{21.5}k\log\log k}\right )$ pairs with congestion at most $14$. 
Since we lose an additional $O(\log^2k)$ factor on the number of pairs routed due to the partitioning step that ensures flow-well-linkedness of the terminals in Section~\ref{subsec: preprocessing}, our algorithm routes $\Omega\left (\frac{\opt}{\log^{23.5}k\log\log k}\right )$ pairs with congestion at most $14$ w.h.p.


\paragraph{Acknowledgements} The author thanks Matthew Andrews and Sanjeev Khanna for many inspiring discussions about the problem.


\bibliography{EDP}
\bibliographystyle{alpha}

\appendix
\section{Table of Parameters}
\renewcommand{\arraystretch}{1.4}

\begin{tabular}{|l|l|p{10cm}|} \hline
$\gkrv(k)$&$\Theta(\log^2k)$&Parameter from the cut-matching game of~\cite{KRV}, from Theorem~\ref{thm: KRV}. Is also denoted by $\gamma$\\ \hline
$\alphasc(k)$&$O(\sqrt{\log k})$& Approximation factor of the algorithm of~\cite{ARV} for Sparsest Cut.\\ \hline
$\alpha(k)$&$\frac{1}{2^{11}\cdot \gkrv(k)\cdot\log k}=\Omega\left (\frac 1 {\log^3k}\right )$&Well-linkedness parameter from Theorem~\ref{thm: oracle-well-linked}\\ \hline
$\alphaWL(k)$&$\alpha(k)/\alphasc(k)=\Omega\left (\frac 1 {\log^{3.5}k}\right )$&Well-linkedness parameter from Corollary~\ref{corollary: weak well linked}\\ \hline
$\beta(k)$&$\Theta(\log k)$&Flow-cut gap for concurrent flow on $k$ terminals\\ \hline
$k_1$&$\frac{k}{192\gamma^3\log \gamma}=\Omega \left (\frac k {\log^6k\log\log k}\right )$&Parameter from the definition of legal contracted graphs\\ \hline
$p$&$\frac{8\beta(k)}{\alphawl(k)}=O(\log^{4.5}k)$&Grouping parameter for the sets $\Gamma_j$\\ \hline
$k'$&$\frac{1}{2\gamma^3}\cdot\lfloor \frac{k_1}{6p}\rfloor=\Omega\left (\frac{k}{\log^{16.5}k\log\log k}\right )$&Number of vertices in the expander $X$\\ \hline
$k^*$&$\lfloor\frac{k_1}{6p}\rfloor$&Size of sets $\Gamma_j'$ (that contain at most one edge from each group of $\gset_j$)\\ \hline
\end{tabular}

\section{Proof of Theorem~\ref{thm: vertex-disjoint routing on expanders}}
\label{sec: appendix-expander-routing}
 
 Let $\ell=4d\beta(n)$, where $\beta(n)=O(\log n)$ is the flow-cut gap for undirected graphs. The algorithm greedily selects a source-sink pair $(s_i,t_i)$ that has a path $P$ of length at most $\ell$ connecting $s_i$ to $t_i$ in the current graph $G$. We then remove all vertices of $P$ from the graph $G$ and continue. The algorithm terminates when for each remaining source-sink pair $(s_i,t_i)$, every path connecting $s_i$ to $t_i$ has length at least $\ell$. 
 
 Note that in each iteration of the algorithm, we route one demand pair, and remove at most $(\ell+1)d$ edges from the graph. The key to the algorithm analysis is to show that when the algorithm terminates, we have removed many edges from the graph, and therefore we have routed many of the demand pairs.
 
 Let $E'$ be the subset of edges removed from the graph by the algorithm, and let $E''$ be the subset of remaining edges. We first claim that there is a multicut in graph $G$ whose value is at most $|E'|+|E''|\cdot \beta(n)/\ell$. Indeed, let $G'=G[E'']$ be the graph obtained when the algorithm terminates, and let $\mset'$ be the set of the surviving source-sink pairs. Consider the instance of the multicut problem on graph $G'$ with the set $\mset'$ of demand pairs. Setting the weight of each edge in $E''$ to $1/\ell$, we obtain a feasible fractional solution to this multicut instance, since the length of every path connecting every pair of terminals is at least $\ell$. Therefore, there is an integral solution to this multicut instance of value $|E''|\cdot \beta(n)/\ell$. Adding the subset $E'$ of edges, we obtain a feasible solution to the multicut problem on the original graph $G$ of value $|E'|+|E''|\cdot \beta(n)/\ell$.
 
 On the other hand, the value of any multicut on graph $G$ is at least $|V|/4$. Indeed, if $E^*$ is any feasible solution to the multicut problem, then each connected component $C$ of $G\setminus E^*$ contains at most $|V|/2$ vertices, and therefore has at least $|V(C)|/2$ out-going edges. Since each edge is counted at most twice, we get that $|E^*|\geq |V|/4$.

We conclude that $|E'|+|E''|\cdot \beta(n)/\ell\geq |V|/4$, and so

\[
|E'|\geq \frac{|V|}{4}-\frac{|E''|\cdot \beta(n)}{\ell}
\geq \frac{|V|}{4}-\frac{|E|\cdot \beta(n)}{\ell} 
\geq \frac{|V|}{4}-\frac{d|V|\beta(n)}{2\ell}
\geq \frac{|V|}{8}\]
 
 since $\ell=4d\beta(n)$. Therefore, at least $|V|/8$ edges have been deleted from the graph. Since in each iteration we only delete at most $d(\ell+1)$ edges, overall the number of pairs routed is at least $\frac{|V|}{8d(\ell+1)}=\Omega\left(\frac{|V|}{d^2\log n}\right )$.

\section{Proof of Theorem~\ref{thm: starting point}}\label{sec: proof of thm starting point}

For the proof of the theorem, we need a more general definition of flow well-linkedness, that was used in~\cite{CKS}.
Suppose we are given a graph $G=(V,E)$, and for each vertex $v\in V$, we are  given a weight $\pi(v)$. For a subset $S\sse V$ of vertices, let $\pi(S)=\sum_{v\in S}\pi(v)$. We say that $G$ is $\pi$-flow well-linked, iff each pair $(u,v)$ of vertices can simultaneously send $\frac{\pi(u)\cdot \pi(v)}{\pi(V)}$ flow units to each other with no congestion.
We start with the following theorem, that was proved in~\cite{CKS}, using a flow-well-linked graph decomposition.

\begin{theorem}[Theorem 2.1 in~\cite{CKS}]
Let $G=(V,E)$ be any graph and let $\mset$ be a set of $k$ source-sink pairs in $G$. 
We can efficiently find a partition $G_1,\ldots,G_{\ell}$ of $G$ into vertex-disjoint induced subgraphs, and for each $1\leq i\leq \ell$, find  a weight function $\pi_i:V(G_i)\rightarrow \reals^+$, with the following properties. Let $\mset'_i\sse \mset$ be the set of source-sink pairs contained in $G_i$, and let $\tset'_i$ be the set of all terminals participating in $\mset_i'$. Then:
\begin{itemize}
\item For all $1\leq i\leq \ell$:
\begin{itemize}
\item for all $u\in \tset'_i$, $\pi_i(u)\leq 1$.
\item for all $(u,v)\in \mset_i'$, $\pi_i(u)=\pi_i(v)$.
\item Graph $G_i$ is $\pi_i$-flow well-linked.
\end{itemize}
\item $\sum_{i=1}^{\ell}\pi_i(\tset'_i)=\Omega(\opt/(\beta(k)\cdot \log \opt))=\Omega(\opt/\log^2k)$.
\end{itemize}
\end{theorem} 

 In order to complete the proof of the theorem, it is enough to show that we can find, for each $1\leq i\leq \ell$, a subset $\mset_i\sse \mset_i'$ of source-sink pairs, with $|\mset_i|=\Omega(\pi_i(\tset_i'))$, such that the set $\tset_i$ of all terminals participating in pairs in $\mset_i$ is flow-well-linked in $G_i$.

Fix some $1\leq i\leq \ell$. We find a grouping $\gset_i$ of the terminals in set $\tset_i'$, using the weights $\pi_i$ and the grouping parameter $p=2$, as in Theorem~\ref{thm: grouping}, so for each group $U\in \gset_i$, $2\leq \pi_i(U)\leq 6$. Next, we will gradually construct the set $\mset_i$ of source-sink pairs, starting from $\mset_i=\emptyset$. In each iteration, we will add one source-sink pair to $\mset_i$, and remove some source-sink pairs from $\mset_i'$, charging their weights to the pair that was added to $\mset_i$. While $\mset_i'$ is non-empty, we perform the following procedure:

\begin{itemize}
\item Let $(s,t)\in \mset_i'$ be any source-sink pair. Add $(s,t)$ to $\mset_i$.

\item If both $s$ and $t$ belong to the same group $U\in \gset$, then for each pair $(u,v)\in \mset_i'$, where $u\in U$ or $v\in U$, remove $(u,v)$ from $\mset_i'$, and charge the weight $\pi_i(v)$ and $\pi_i(v)$ to $(s,t)$. Notice that the total weight charged to $(s,t)$ is at most $12$.

\item Otherwise, let $U_1$ be the group to which $s$ belongs, and let $U_2$ be the group to which $t$ belongs.
For each pair $(u,v)\in \mset_i'$, such that either $u\in U_1\cup U_2$, or $v\in U_1\cup U_2$, remove $(u,v)$ from $\mset_i'$, and charge the weights $\pi_i(u)$ and $\pi_i(v)$ to $(s,t)$. Notice that the total weight charged to $(s,t)$ in this step is at most $24$.
\end{itemize}
The procedure stops when $\mset_i'=\emptyset$. Let $\mset_i$ be the resulting set of source-sink pairs, and let $\tset_i$ be the set of terminals participating in them. From the above charging scheme, it is clear that $|\mset_i|=\Omega(\pi_i(\tset_i'))$, as required. Observe also that for each group $U\in \gset_i$, at most one terminal $v\in U$ belongs to $\tset_i$. Finally, we need to show that $G_i$ is flow well-linked for $\tset_i$. For each vertex $v\in \tset_i$, let $U_v\in \gset_i$ be the group to which $v$ belongs. 

Suppose we are given any matching $\mset^*$ on the set $\tset_i$ of terminals. We show how to route this matching with congestion at most $2$ in $G$. We do so in two steps. In the first step, we construct a flow $F_1$, where for each pair $(v,v')\in \mset^*$, the vertices in $U_v$ send 1 flow unit in total to the vertices in $U_{v'}$, each vertex $x\in U_v$ sends at most $\pi_i(x)$ flow units and each vertex $y\in U_{v'}$ receives at most $\pi_i(y)$ flow units, with total congestion at most $1$. This flow is defined as follows. Recall that graph $G_i$ is $\pi_i$-well-linked. Therefore, every pair $(x,y)$ of vertices can send $\frac{\pi_i(x)\cdot \pi_i(y)}{\pi_i(V(G_i))}$ flow units to each other with no congestion. Let $F$ denote this flow. Fix some pair $(v,v')\in \mset^*$. In flow $F$, there are $\pi_i(U_v)$ flow units originating from the vertices in $U_v$, that are then distributed among the vertices of $G$, and the amount of flow each vertex $z$ of $G$ receives is $\pi_i(z)\cdot \pi_i(U_v)/\pi_i(V(G_i))$. If $\pi_i(U_v)>2$, we scale the flow originating from vertices in $U_v$ down by factor $\pi_i(U_v)/2$, so that every vertex $z$ of $G$ now receives $2\pi_i(z)/\pi_i(V(G_i))$ flow units from $U_v$. We perform a similar transformation for the flow originating at the vertices of $U_{v'}$, and we concatenate both flows. As a result, we obtain a flow where the vertices in $U_v$ send two flow units in total to the vertices in $U_{v'}$. Taking the union of these flows over all $(v,v')\in \mset^*$, and scaling them down by factor $2$, gives us the flow $F_1$. It is easy to see that the total congestion caused by $F_1$ is at most $1$. This is since each flow-path in $F$ is used at most twice: once for each of its end-points.
Finally, in order to route the matching $\mset^*$, consider any pair $(v,v')\in \mset^*$. Vertex $v$ will distribute one flow unit to the vertices in $U_v$, along the tree $T_{U_v}$, where the amount of flow each vertex $x\in T_{U_v}$ receives equals to the amount of flow it sends out in $F_1$. We then use the flow $F_1$ to route this one flow unit to the vertices of $U_{v'}$. Finally, vertex $v'$ collects one flow unit from the vertices of $U_{v'}$ along the tree $T_{U_{v'}}$. It is easy to see that the total congestion caused by this flow is at most $2$, since all trees $\set{T_{U}}_{U\in \gset_i}$ are edge-disjoint.

\end{document}